\newif\ifarxiv
\arxivtrue

\ifarxiv{}
\documentclass[a4paper,USenglish,cleveref, autoref,pagebackref,
thm-restate,authorcolumns]{lipics-v2019}
\nolinenumbers
\hideLIPIcs

\else{}
\documentclass[review,onefignum,onetabnum]{siamart190516}

\usepackage{epstopdf}
\ifpdf
  \DeclareGraphicsExtensions{.eps,.pdf,.png,.jpg}
\else
  \DeclareGraphicsExtensions{.eps}
\fi
\newsiamremark{remark}{Remark}
\newsiamremark{hypothesis}{Hypothesis}
\newsiamremark{construction}{Construction}
\crefname{hypothesis}{Hypothesis}{Hypotheses}
\crefname{subsection}{Section}{Sections}
\crefname{section}{Section}{Sections}
\newsiamthm{claim}{Claim}
\newsiamthm{observation}{Observation}
\fi{}

\usepackage{lipsum}
\usepackage{amsfonts}
\usepackage{graphicx}
\usepackage{algorithmic}

\usepackage{tabularx}
\usepackage{tikz}
\tikzstyle{vertex}=[circle,fill=black,inner sep=1.5,draw,minimum size=.1cm]
\tikzstyle{vertex2}=[circle,fill=white,inner sep=1.5,draw,minimum size=.25cm,thick]
\tikzstyle{edge}=[]

\usepackage{subcaption}
\usepackage[algo2e,ruled,vlined,linesnumbered,algosection,hanginginout]{algorithm2e}
\SetEndCharOfAlgoLine{}
\usepackage{paralist}

\newcommand{\G}{\mathcal G}
\newcommand{\M}{\mathcal M}
\newcommand{\Q}{\mathcal Q}
\newcommand{\T}{\mathcal T}
\newcommand{\mathS}{\mathcal S}
\newcommand{\yes}{\emph{yes}}
\newcommand{\no}{\emph{no}}
\newcommand{\dgg}[2]{\ensuremath{\widehat{Z}_{{#1}, {#2}}}}
\newcommand{\MaxTempMatchingSize}{\mu}

\newcommand{\optproblemdef}[3]{
\begin{center}   
    \fbox{~\begin{minipage}{.9\textwidth}
      \vspace{2pt} 
     
      \noindent
      \normalsize\textsc{#1}
      
      \vspace{4pt}
      \setlength{\tabcolsep}{3pt}
      \renewcommand{\arraystretch}{1.0}
      \begin{tabularx}{\textwidth}{@{}lX@{}}
	\normalsize\textbf{Input:} 	& \normalsize#2 \\
	\normalsize\textbf{Output:} 	& \normalsize#3
      \end{tabularx}
    \end{minipage}}
    \end{center}
       \vspace{2pt}
}

\ifarxiv{}
\theoremstyle{plain}
\newtheorem{observation}[theorem]{Observation}
\theoremstyle{definition}
\newtheorem{construction}{Construction}

		\title{Computing Maximum Matchings in Temporal Graphs}

\newcommand{\tuaddress}{Technische Universit\"at Berlin, Algorithmics and Computational Complexity, Berlin, Germany}

\title{Computing Maximum Matchings in Temporal Graphs} %

\author{George B. Mertzios}
{Department of Computer Science, Durham University, UK}
{george.mertzios@durham.ac.uk}
{https://orcid.org/0000-0001-7182-585X}
{Supported by the EPSRC grant EP/P020372/1.}

\author{Hendrik Molter}
{\tuaddress}
{h.molter@tu-berlin.de}
{https://orcid.org/0000-0002-4590-798X}
{Supported by the DFG, project MATE (NI369/17).}

\author{Rolf Niedermeier}
{\tuaddress}
{rolf.niedermeier@tu-berlin.de}
{https://orcid.org/0000-0003-1703-1236}
{}

\author{Viktor Zamaraev}{Department of Computer Science, University of Liverpool, UK}
{viktor.zamaraev@liverpool.ac.uk}
{https://orcid.org/0000-0001-5755-4141}
{Supported by the EPSRC grant EP/P020372/1. The main part of this paper was prepared while affiliated with the Department of Computer Science, Durham University, UK.}

\author{Philipp Zschoche}
{\tuaddress}
{zschoche@tu-berlin.de}
{https://orcid.org/0000-0001-9846-0600}
{}
\authorrunning{G.\ B.\ Mertzios, H.\ Molter, R.\ Niedermeier, V.\ Zamaraev, and P.\ Zschoche}%
\Copyright{George B.\ Mertzios, Hendrik Molter, Rolf Niedermeier, Viktor Zamaraev, and Philipp Zschoche}%
\category{}%
\supplement{}%
\acknowledgements{}%

\authorrunning{G. B. Mertzios, H. Molter, R. Niedermeier, V. Zamaraev, and P. Zschoche}%

\Copyright{George B. Mertzios, Hendrik Molter, Rolf Niedermeier, Viktor Zamaraev, and Philipp Zschoche}%

\else{}

\headers{Computing Maximum Matchings in Temporal Graphs}{G.\ B.\ Mertzios, H.\ Molter, R.\ Niedermeier, V.\ Zamaraev, and P.\ Zschoche}

\title{Computing Maximum Matchings in Temporal Graphs\thanks{An extended abstract of this work appears in the proceedings of STACS 2020 \cite{us}. This version provides full proof details. The main part of this paper was prepared while VZ was affiliated with the Department of Computer Science, Durham University, UK.}
		\funding{Supported by the EPSRC (grant EP/P020372/1) 
and the DFG (project MATE (NI369/17)).}}

\author{George B. Mertzios\thanks{Department of Computer Science, Durham University, UK (\email{george.mertzios@durham.ac.uk})}
\and
Hendrik Molter\thanks{Technische Universit\"at Berlin, Faculty~IV, Algorithmics
and Computational~Complexity, Berlin, Germany (\email{h.molter@tu-berlin.de}, \email{rolf.niedermeier@tu-berlin.de}, \email{zschoche@tu-berlin.de})}
\and
Rolf Niedermeier\footnotemark[3]
\and
Viktor~Zamaraev\thanks{Department of Computer Science, University of Liverpool, UK (\email{viktor.zamaraev@liverpool.ac.uk})}
\and
Philipp Zschoche\footnotemark[3]
}

\usepackage{amsopn}

\ifpdf
\hypersetup{
  pdftitle={Computing Maximum Matchings in Temporal Graphs},
  pdfauthor={G.\ B.\ Mertzios, H.\ Molter, R.\ Niedermeier, V.\ Zamaraev, and P.\ Zschoche}
}
\fi

\fi{}
\begin{document}

\maketitle

\begin{abstract}
Temporal graphs are graphs whose topology 
is subject to discrete changes over time. Given a static underlying
graph $G$, a temporal graph is represented by assigning a set of integer 
time-labels to every edge $e$ of $G$, indicating the discrete time steps at 
which $e$ is active. 
We introduce and study the complexity of a natural temporal 
extension of the classical graph problem \textsc{Maximum Matching}, 
taking into account the dynamic nature of temporal graphs. 
In our problem, \textsc{Maximum Temporal Matching}, we are looking for the 
largest possible number of time-labeled edges (simply \emph{time edges}) 
$(e,t)$ such that no vertex is matched more than once within any time 
window of~$\Delta$~consecutive time slots, where $\Delta \in \mathbb{N}$ 
is given. 
The requirement that a vertex cannot be matched twice in any $\Delta$-window 
models some necessary ``recovery'' period that needs 
to pass for an entity (vertex) after being paired up for some activity with another entity. 
We prove strong computational hardness results for \textsc{Maximum Temporal Matching}, even for elementary cases. 
To cope with this computational hardness, we mainly focus on fixed-parameter algorithms with respect to natural parameters, as well as on polynomial-time approximation algorithms. 
\end{abstract}

\ifarxiv{}\fi{}

\section{Introduction}
Computing a maximum matching in an undirected graph (a maximum-cardinality set of ``independent edges'', 
i.e., edges which do not share any endpoint) is one of the most fundamental graph-algorithmic primitives.
In this work, we lift the study of the algorithmic complexity of computing maximum matchings from static graphs
to the---recently strongly growing---field of \emph{temporal
graphs}~\cite{akridaGMS16,Stochastic_Temporal_AkridaMSZ19,AkridaMSZ18,AxiotisF16,BHMMNS18,ErlebachHK15,mertziosMCS19,MMZ2019}.
In a nutshell, a temporal graph is a graph whose topology is subject to discrete changes over time. 
We adopt a simple and natural model for temporal graphs which originates from
the foundational work of Kempe et al.~\cite{kempe}.
According to this model, every edge of a static graph is given along with a set of time labels, while the vertex set remains unchanged.

\begin{definition}[Temporal Graph]
\label{temp-graph-def} 
A \emph{temporal graph} $\G=(G,\lambda)$ is a pair $(G,\lambda)$,
where $G=(V,E)$ is an underlying (static) graph and $\lambda :E\rightarrow
2^{\mathbb{N}}\setminus\{\emptyset\}$ is a \emph{time-labeling} function that specifies which edge is \emph{active} at what time.
\end{definition}
An alternative way to view a temporal graph is to see it as an ordered set (according to the discrete time slots) 
of graph instances (called \emph{snapshots}) on a fixed vertex set. 
Due to their vast applicability in many areas, temporal graphs 
have been studied from different perspectives under various names 
such as \emph{time-varying}~\cite{TangMML10-ACM}, %
\emph{evolving}~\cite{clementi,Ferreira-MANETS-04}, %
\emph{dynamic}~\cite{CasteigtsFloccini12,GiakkoupisSS14}, 
and \emph{graphs over time}~\cite{Leskovec-Kleinberg-Faloutsos07}; 
see also the survey papers~\cite{flocchini1,flocchini2,CasteigtsFloccini12} and
the references therein.

In this paper we introduce and study the complexity of a natural temporal extension of the classical problem \textsc{Maximum Matching}, 
which takes into account the dynamic nature of temporal graphs. 
To this end, we extend the notion of ``edge independence'' to the temporal setting: 
two time-labeled edges (simply \emph{time edges}) $(e,t)$ and $(e',t')$ are \emph{$\Delta$-independent} 
whenever (i)~the edges $e,e'$ do not share an endpoint or (ii)~their time labels $t,t'$ are at least~$\Delta$ time units apart 
from each other.\footnote{Throughout the paper, $\Delta$ always refers to that number, 
and never to the maximum degree of a static graph (which is another common use of $\Delta$).}
Then, for any given $\Delta$, the problem \textsc{Maximum Temporal Matching} asks for the largest possible set 
of pairwise $\Delta$-independent edges in a temporal graph. 
That is, in a feasible solution, no vertex can be matched more than once within any time window of length~$\Delta$. 
This requirement can model settings where a short ``recovery'' period is needed for every vertex that participates in the matching, 
e.g.,~a short rest after an energy-demanding activity. 
Thus it makes particular sense to study the complexity of the problem for
small (constant) values of $\Delta$.
Note that the concept of $\Delta$-windows has also been employed in other
temporal graph problem settings~\cite{AkridaMSZ18,CHMZ19,HimmelMNS17,MMZ2019}.

Our main motivation for studying \textsc{Maximum Temporal Matching} is of theoretical nature, 
namely to lift one of the most classical optimization problems, \textsc{Maximum Matching}, to the temporal setting. 
As it turns out, \textsc{Maximum Temporal Matching} is computationally hard to approximate: 
we prove that the problem is APX-hard, even when $\Delta=2$ and the lifetime $T$ of the temporal graph 
(i.e.,~the maximum edge label) is $3$ (see Section~\ref{sec:apx}). 
That is, unless P=NP, there is no Polynomial-Time Approximation Scheme (PTAS) for any $\Delta\geq 2$ and $T\geq 3$. 
In addition, we show that the problem remains NP-hard even if the underlying graph $G$ is just a path (see Section~\ref{sec:pathhardness}). 
Consequently, we mainly turn our attention to approximation and to fixed-parameter algorithms (see \cref{sec:algos}). 
In order to prove our hardness results, we introduce the notion of
\emph{temporal line graphs} which form a class of (static) graphs of independent
interest and may prove useful in other contexts, too.
This notion paves the way to reduce \textsc{Maximum Temporal Matching} to the
problem of computing a large independent set in a static graph (i.e.,~in the temporal line graph that is defined from the input temporal graph). 
Moreover, as an intermediate result, we show (see Theorem~\ref{prop:nphgridis}) that the classic problem 
\textsc{Independent Set} (on static graphs) remains 
NP-hard on induced subgraphs of \emph{diagonal grid} graphs, thus 
strengthening an old result of Clark et al.~\cite{clark1990unit} for unit disk graphs.

During the last few decades it has been repeatedly observed that 
for many variations of \textsc{Maximum Matching} it is straightforward 
to obtain online (resp.\ greedy offline approximation) algorithms which achieve 
a competitive (resp.\ an approximation) ratio of $\frac{1}{2}$, 
while great research efforts have been made to increase the ratio 
to $\frac{1}{2} + \varepsilon$, for \emph{any} constant $\varepsilon >0$. 
Originating in the foundational work of Karp et al.~\cite{karp1990optimal} 
on the randomized online algorithm \textsc{Ranking} for the \textsc{Online
Bipartite Matching} problem, there has been a long line of recent research on
providing a sequence of $(\frac{1}{2} + \varepsilon)$-competitive algorithms for
many different variations of \textsc{Online Matching}, see
e.g.~\cite{BuchbinderST2019,gamlath2019-arxiv,huang2018match,huang2019tight}.
This difficulty of breaking the barrier of the ratio $\frac{1}{2}$ also appears
in offline variations of the \textsc{Matching} problem.
It is well known that an arbitrary greedy algorithm for \textsc{Matching} gives approximation ratio at least $\frac{1}{2}$~\cite{HK78,KorteH78}, 
while it remains a long-standing open problem to determine how well a randomized greedy algorithm can perform. 
Aronson et al.~\cite{aronson1995randomized} provided the so-called Modified Randomized Greedy (MRG) algorithm 
which approximates the maximum matching within a factor of at least $\frac{1}{2} + \frac{1}{400,000}$. 
Recently, Poloczek and Szegedy~\cite{poloczek2012randomized} proved that MRG actually provides 
an approximation ratio of $\frac{1}{2} + \frac{1}{256}$. 
Similarly to the above problems, it is straightforward\footnote{To achieve the straightforward $\frac{1}{2}$-approximation 
it suffices to just greedily compute at every time slot a maximal matching among the edges that are $\Delta$-independent 
with the current solution.} %
to approximate \textsc{Maximum Temporal Matching} in polynomial time within a factor of $\frac{1}{2}$. 
However, we manage to
provide a simple (non-randomized) approximation algorithm which, for any
constant~$\Delta$, achieves an approximation ratio $\frac{1}{2} + \varepsilon$ for some $\varepsilon = \varepsilon(\Delta)$. For $\Delta=2$ this ratio is $\frac{2}{3}$, 
while for an arbitrary constant $\Delta$ it becomes 
$\frac{\Delta}{2\Delta-1} = \frac{1}{2} + \frac{1}{2(2\Delta-1)}$ (see Section~\ref{sec:approxAlg}).

Apart from polynomial-time approximation algorithms, the classical (static)
matching problem (which is polynomial-time solvable) has recently also attracted 
many research efforts in the area of parameterized algorithms 
for polynomial-time solvable problems. Parameters which have been studied
include the solution size~\cite{GiannopoulouMN15}, the modular-width~\cite{KratschNelles18}, the clique-width~\cite{CoudertDP18}, 
the treewidth~\cite{FominLSPW-18}, the feedback vertex 
number~\cite{MertziosNN17}, and the feedback edge
number~\cite{KorenweinNNZ18,MertziosNN17}.
Given that \textsc{Maximum Temporal Matching} is NP-hard, 
we show fixed-parameter tractability with respect to 
the solution size as a parameter.
Finally, we show fixed-parameter tractability
with respect to the combined parameter of time-window size $\Delta$ and the size
of a maximum matching of the underlying graph (which may be significantly
smaller than the cardinality of a maximum temporal matching of the temporal graph). 
Our algorithmic techniques are essentially based on efficient and effective data
reduction (kernelization) and matroid theory (see \cref{sec:algos}).

It is worth mentioning that another temporal variation 
of \textsc{Maximum Matching}  
was recently proposed by Baste~et~al.~\cite{baste2018temporal}. 
The main difference to our model is that their model requires edges to exist 
in at least $\Delta$~\emph{consecutive} snapshots 
in order for them to be eligible for a matching. 
Thus, their matchings need to consist of time-consecutive edge blocks, which requires some data cleaning
on real-word instances in order to perform meaningful experiments~\cite{baste2018temporal}. %

It turns out that the model of Baste et al.~is a special case of our model, as there is an easy reduction from their model to ours, 
and thus their results are also implied by ours. 
Baste et al.~\cite{baste2018temporal} showed that solving 
(using their definition) \textsc{Maximum Temporal Matching} is NP-hard 
for $\Delta\geq 2$. 
In terms of parameterized complexity, they provided 
a polynomial-sized kernel for the combined parameter~$(k, \Delta)$, 
where~$k$ is the size of the desired solution. 

To the best of our knowledge, the main alternative model for temporal matchings in temporal graphs 
is the concept of multistage (perfect) matchings which was introduced by Gupta~et~al.~\cite{gupta2014changing}. 
This model, which is inspired by reconfiguration or reoptimization problems, is not directly related to ours: 
roughly speaking, their goal is to find perfect matchings for every snapshot of
a temporal graph such that the matchings only \emph{slowly change} over time.
In this setting one mostly encounters computational intractability, which leads 
to several results on approximation hardness and
approximation algorithms~\cite{bampis2018multistage,gupta2014changing}.

\section{Preliminaries}
\label{sec:preliminaries}
In this section we present all necessary notation and terminology as well as some easy initial observation about our problem setting.

\subsection{Notation and Terminology}
Let $\mathbb{N}$ denote the natural numbers without zero. 
We refer to a set of consecutive natural numbers $[i,j]=\{i,i+1,\ldots ,j\}$ for some $i,j\in \mathbb{N}$ with $i\leq j$ as an \textit{interval}, and to the number~$j-i+1$ as the \textit{length} of the interval.
If $i=1$, then we denote $[i,j]$ simply by~$[j]$.
By~$\mathbb F_p$ we denote the finite field on $p$~elements.
For the sake of brevity, the notation $A \uplus B$ denotes the union of two sets $A$ and $B$ and implicitly indicates that the sets are disjoint.
We call a family of sets \(Z_1,\dots,Z_\ell\) a \emph{partition} of a set~\(A\)
if~\(Z_1\uplus\dots\uplus Z_\ell=A\)
and \(Z_i\ne\emptyset\) for each~\(i\in\{1,\dots,\ell\}\).
A \textit{$p$-family} is a family of sets where each set is of size exactly~$p$.
\ifarxiv{}
\subparagraph*{Static graphs.}
\fi{}
We use standard notation and terminology from graph theory~\cite{diestel2000graphentheory}. 
Given an undirected (static) graph $G=(V,E)$ with $E\subseteq \binom{V}{2}$, we denote by~$V(G)=V$ and $E(G)=E$ the sets of its
vertices and edges, respectively. 
We call two vertices $u,v\in V$ \emph{adjacent} if $\{u,v\}\in E$. We call two edges $e_1,e_2\in E$ \emph{adjacent} if $e_1\cap e_2\neq \emptyset$.
By $P_n$ we denote a graph that is a path with $n$ vertices. By $\nu(G)$ we denote the size of a maximum matching in $G$. Whenever it is clear from the context, we omit $G$.
Two graphs $G_1=(V_1,E_1)$ and $G_2=(V_2,E_2)$ are \emph{isomorphic} if there is a bijection $\sigma : V_1 \rightarrow V_2$ such that for all $u,v\in V_1$ we have that $\{u,v\}\in E_1$ if and only if $\{\sigma(u),\sigma(v)\}\in E_2$.
Given a graph $G=(V,E)$ and an edge $\{u,v\}\in E$, \emph{subdividing} the edge $\{u,v\}$ results in a graph isomorphic to $G'=(V', E')$ with $V'=V\cup\{w\}$ for some $w\notin V$ and $E'=(E\setminus \{\{u,v\}\})\cup\{\{v,w\},\{u,w\}\}$. A graph $H$ is a \emph{subdivision} of a graph $G$ if there is a sequence of graphs $G_1, G_2, \ldots, G_x$ with $G_1=G$ such that for each $G_i=(V_i, E_i)$ with $i<x$ there is an edge $e\in E_i$ and subdividing $e$ results in a graph isomorphic to $G_{i+1}$, and~$G_x$~is isomorphic to $H$. 
A graph $H$ is a \emph{topological minor} of $G$ if there is a subgraph $G'$ of $G$ that is a subdivision of $H$. 
A graph $H$ is an \emph{induced topological minor} of $G$ if there is an \emph{induced} subgraph $G'$ of $G$ that is a subdivision of $H$. %
A \emph{line graph} of a (static) graph $G=(V,E)$ is the graph~$L(G)$ with~$V(L(G))=\{v_e\mid e\in E\}$ and for all $v_e, v_{e'}\in V(L(G))$ we have that~$\{v_e,v_{e'}\}\in E(L(G))$ if and only if $e\cap e'\neq \emptyset$~\cite{diestel2000graphentheory}. Recall that a \emph{maximum independent set} of a (static) graph $G=(V,E)$ is a vertex set~$V'\subseteq V$ of maximum cardinality such that for all $u,v\in V'$ we have that $\{u,v\}\notin E$. In the context of matchings, line graphs are of special interest since the cardinality of a maximum matching in a graph equals the cardinality of a maximum independent set in its line graph. Indeed, a matching in a graph can directly be translated into an independent set in its line graph and vice versa~\cite{diestel2000graphentheory}.

\ifarxiv{}
\subparagraph*{Parameterized complexity.} 
\fi{}
We use standard notation and terminology from parameterized
complexity~\cite{cygan2015parameterized,downey2013fundamentals}. %
A \emph{parameterized problem} is a language $L\subseteq \Sigma^* \times \mathbb{N}$, where $\Sigma$ is a finite alphabet. We call the second component
the \emph{parameter} of the problem.
A parameterized problem is \emph{fixed-pa\-ram\-e\-ter tractable} (in the complexity class FPT)
if there is an algorithm that solves each instance~$(I, r)$ in~$f(r) \cdot |I|^{O(1)}$ time,
for some computable function $f$. 
If a parameterized problem $L$ is NP-hard for a constant parameter value, it cannot be contained in FPT\footnote{It cannot even be contained in the larger parameterized complexity class XP unless P$\ =\ $NP.} unless P$\ =\ $NP.
A parameterized problem $L$ admits a \emph{polynomial kernel} if there is a polynomial-time algorithm that transforms each instance $(I,r)$ into an instance $(I', r')$ such that $(I,r)\in L$ if and only if $(I',r')\in L$ and~$|(I', r')|\le r^{O(1)}$. 

\ifarxiv{}
\subparagraph*{Temporal graphs.}
\fi{}

Throughout the paper we consider temporal graphs $\G$ with \emph{finite
lifetime}~$T(\G):=\max \{t\in \lambda (e)\mid e\in E\}$, that is, there is a maximum label assigned by~$\lambda$ to an
edge of~$G$. When it is clear from the context, we denote the lifetime of $\G$ simply by $T$.
The \emph{snapshot} (or \emph{instance}) of $\G$ 
\emph{at time}~$t$ is the static graph $G_{t}=(V,E_{t})$, where $%
E_{t}:=\{e\in E\mid t\in \lambda (e)\}$. 
We refer to each integer $t\in [T]$ as a \emph{time slot} of $\G$. 
For every $e\in E$ and every time slot $t\in\lambda(e)$, we denote the \emph{appearance
of edge} $e$ \emph{at time} $t$ by the pair $(e,t)$, which we also call a \emph{time edge}.
We denote the set of edge appearances of a temporal graph $\G =(G=(V,E),\lambda)$ by $\mathcal E(\G) := \{ (e,t) \mid e \in E$ and $t \in \lambda(e) \}$.
For every $v\in V$ and every time slot $t$, we denote the \emph{appearance
of vertex} $v$ \emph{at time} $t$ by the pair $(v,t)$. 
That is, every vertex~$v$ has $T$ different appearances (one for each time slot) during the
lifetime of $\G$. 
For every time slot $t\in [T]$, we denote by $V_{t}:=\{(v,t) : v\in V\}$ the set
of all vertex appearances of~$\G$ at time slot~$t$. 
Note that the set of all vertex appearances in $\G$ is $V\times [T] = \bigcup_{1\leq t\leq T} V_t$.
Two vertex appearances $(v,t)$ and $(w,t)$ are \emph{adjacent} if the temporal graph has the time edge $(\{v,w\},t)$.
For a temporal graph $\G= (G,\lambda)$ and a set of time edges $M$, 
we denote by~$\G \setminus M := (G', \lambda')$ the temporal graph $\G$ without the time edges in $M$, 
where $\lambda'(e) := \lambda(e) \setminus \{ t \mid (e,t) \in M \}$ for all $e\in E'$, and
$G':=(V, E')$ with $E':=\{e\in E\mid \lambda'(e) \neq \emptyset\}$.
For a subset $S\subseteq [T]$ of time slots and 
a time edge set $M$, we denote by~$M|_{S} := \{ (e,t) \in M \mid t\in S \}$ the set of time edges in~$M$ with a label in~$S$.
For a temporal graph $\G$, we denote by $\G|_{S} := \G \setminus (\mathcal E(\G)|_{[T] \setminus S})$ the temporal graph 
where only time edges with a label in $S$ are present.

In the remainder of the paper we denote by
$n$ and $m$  
the number of vertices and edges of the underlying graph $G$, respectively, unless
otherwise stated. 
We assume that there is no compact representation of the 
labeling $\lambda$, that is, $\G$ is given with an
explicit list of labels for every edge,
and hence the \emph{size} of a temporal graph $\G$ is~$|\G| := |V|+\sum_{t=1}^{T}|E_{t}| \in O(n+mT)$. 
Furthermore, in accordance with 
the literature~\cite{wu2016efficient,zschocheFMN18} we assume that the lists of labels are given in ascending order.

\ifarxiv{}
\subparagraph*{Temporal matchings.}
\fi{}
A \emph{matching} in a (static) graph $G=(V,E)$ is a set $M\subseteq E$ of edges such that for all $e, e'\in M$ we have that $e\cap e'=\emptyset$. In the following, we transfer this concept to temporal graphs.

For a natural number $\Delta$, 
two time edges $(e,t)$, $(e',t')$ are \emph{$\Delta$-independent} 
if $e \cap e' = \emptyset$ or~$|t - t'| \geq \Delta$. If two time edges are not $\Delta$-independent, then we say that they are \emph{in conflict}.
A time edge $(e,t)$ \emph{$\Delta$-blocks} a vertex appearance~$(v,t')$ 
(or $(v,t')$ is \emph{$\Delta$-blocked} by~$(e,t)$) 
if $v \in e$ and $|t - t'| \leq \Delta -1$.
A~\emph{$\Delta$-temporal matching} $M$ of a temporal graph $\G$ is a set of
time edges of $\G$ which are pairwise $\Delta$-independent. Formally, it is defined as follows.
\begin{definition}[$\Delta$-Temporal Matching]
A \emph{$\Delta$-temporal matching} of a temporal graph $\G$ is a set 
$M$
of time edges of $\G$ such that for every pair of distinct time edges $(e,t),(e',t')$ in $M$ we have that $e \cap e' = \emptyset$ or $|t - t'| \geq \Delta$.
\end{definition}
A $\Delta$-temporal matching is called \textit{maximal} if it is not properly contained in any other $\Delta$-temporal matching.
A $\Delta$-temporal matching is called \textit{maximum} if there is no $\Delta$-temporal matching of larger cardinality.
We denote by $\MaxTempMatchingSize_{\Delta}(\G)$ 
the size of a maximum $\Delta$-temporal matching in $\G$.

Having defined temporal matchings, we naturally arrive at the following central problem.

\medskip

\optproblemdef{Maximum Temporal Matching}
{A temporal graph $\G = (G,\lambda)$ and an integer $\Delta \in \mathbb N$.}
{A $\Delta$-temporal matching in $\G $ of maximum cardinality.}

\medskip

We refer to the problem of deciding whether a given temporal graph admits a $\Delta$-temporal matching of a given size $k$ by \textsc{Temporal Matching}.

We remark that our definition of a $\Delta$-temporal matching is similar to
\emph{$\gamma$-matchings} by Baste~et~al.~\cite{baste2018temporal}. 
We discuss some basic observations about our problem settings in \cref{app:basic} and discuss the relation between our model and the model of Baste et al.~\cite{baste2018temporal} in \cref{relation-to-baste}.

\ifarxiv{}
\subparagraph*{Temporal line graphs.}
\fi{}
In the following, we transfer the concept of line graphs to temporal graphs and temporal matchings.
We make use of temporal line graphs in the NP-hardness result of \cref{sec:pathhardness}.

The \emph{$\Delta$-temporal line graph} of a temporal graph $\G$ is a static
graph that has a vertex for every time edge of $\G$ and two vertices are connected by an edge if the corresponding time edges are in conflict, i.e., they cannot be both part of a $\Delta$-temporal matching of $\G$.
We say that a graph $H$ is a \emph{temporal line graph} if there exists $\Delta$ and a temporal graph $\G$ such that $H$ is isomorphic to the $\Delta$-temporal line graph of $\G$. Formally, temporal line graphs and $\Delta$-temporal line graphs are defined as follows.

\begin{definition}[Temporal Line Graph]
Given a temporal graph $\G=(G=(V,E), \lambda)$ and a natural number~$\Delta$, the \emph{$\Delta$-temporal line graph} $L_\Delta(\G)$ of $\G$ has vertex set
$V(L_\Delta(\G)) := \{ e_t \mid e\in E \wedge t\in\lambda(e)\}$ and edge set
	$E(L_\Delta(\G)) :s= \{\{e_t,e'_{t'}\} \mid e\cap e' \neq \emptyset \wedge
	|t-t'|<\Delta\}$.
We say that a graph $H$ is a \emph{temporal line graph} if there is a temporal graph $\G$ and an integer~$\Delta$ such that $H=L_\Delta(\G)$.
\end{definition}

By definition, $\Delta$-temporal line graphs have the following property.
\begin{observation}\label{obs:TMequivMIS}
Let $\G$ be a temporal graph and let $L_\Delta(\G)$ be its $\Delta$-temporal line graph. The cardinality of a maximum independent set in $L_\Delta(\G)$ equals the size of a maximum $\Delta$-temporal matching of $\G$.
\end{observation}
It follows that solving \textsc{Temporal Matching} on a temporal graph $\G$ is equivalent to solving \textsc{Independent Set} on $L_\Delta(\G)$.

\subsection{Preliminary results and observations}\label{app:basic}
Note that when the input parameter $\Delta$ in \textsc{Maximum Temporal Matching} is equal to 1, the problem
can be solved efficiently, because it reduces to $T$ independent instances of (static) \textsc{Maximum Matching}. 

At the other extreme, there are instances $(\G=(G,\lambda), \Delta,k)$ in which
$\Delta$ coincides with the lifetime~$T$, i.e.,\ $\Delta=T$.
In this case the problem can also be solved in polynomial time. Indeed, a maximum $\Delta$-temporal matching $M$
can be found as follows:
\begin{enumerate}
	\item Find a maximum matching $R$ in the underlying graph $G$;
	\item Initialize $M = \emptyset$. For every edge $e$ in $R$ add in the final solution $M$ exactly one (arbitrary) 
	time edge $(e,t)$, where $t \in \lambda(e)$.
	\item Output $M$.
\end{enumerate}
The time complexity of the above procedure is dominated by the time required to construct the underlying graph $G$ 
and the time needed to find a maximum matching in $G$. 
The former can be done in $O(Tm) = O(\Delta m)$ time.
The latter can be solved in $O(\sqrt{n}m)$ time~\cite{micali1980v}.
Thus, we have the following.

\begin{observation}\label{obs:DeltaEqT}

	Let $\G$ be a temporal graph of lifetime $T$.
	We can compute a~$T$-temporal matching in $\G$ in $O(m (\sqrt{n} + T))$ time.
\end{observation}

Furthermore, it is easy to observe that computational hardness of \textsc{Temporal Matching} for some fixed value of $\Delta$ implies hardness for all larger values of $\Delta$.
This allows us to construct hardness reductions for small fixed values of $\Delta$ and still obtain general hardness results.

\begin{observation}%
\label{obs:D1toD}
	For every fixed $\Delta$, the problem \textsc{Temporal Matching} on instances
	$(\G, \Delta+1, k)$ is computationally at least as hard as \textsc{Temporal
	Matching} on instances $(\G, \Delta, k)$.
\end{observation}
\begin{proof}
	The result immediately follows from  the observation that a temporal graph~$\G$ has a $\Delta$-temporal matching
	of size at least $k$ if and only if the temporal graph $\G'$ has a $(\Delta+1)$-temporal matching of size at least $k$, where
	$\G'$ is obtained from $\G$ by inserting one edgeless snapshot
	after every $\Delta$ consecutive snapshots (see~Figure~\ref{fig:Delta-increase}).
\begin{figure}[t]
	\centering
	\includegraphics[width=\linewidth]{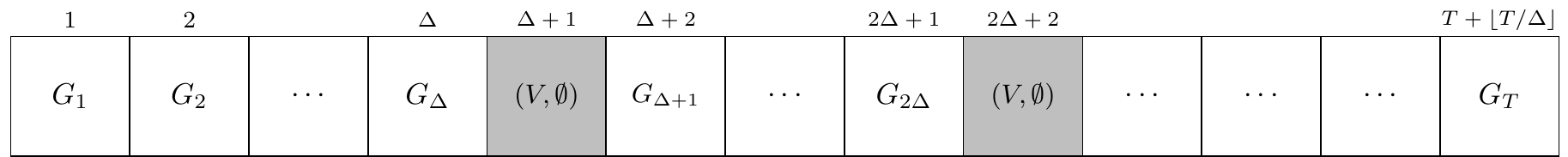}  
	\caption{Inserting ``empty'' snapshots to reduce \textsc{Temporal Matching} on instances $(\G, \Delta, k)$ to 
	\textsc{Temporal Matching} on instances $(\G, \Delta+1, k)$.}
	\label{fig:Delta-increase}
\end{figure}
\end{proof}%

Lastly, it is easy to see that one can check in polynomial time whether a given
set of time edges is a $\Delta$-temporal matching. This implies that
\textsc{Temporal Matching} is contained in NP and in subsequent NP-completeness
statements we will only discuss the hardness part of the proof.

\subsection{Relation to \textsc{$\gamma$-Matching} by Baste et al.~\cite{baste2018temporal}}
\label{relation-to-baste}
We refer to the variant of temporal matching introduced by Baste~et~al.~\cite{baste2018temporal} as \textsc{$\gamma$-Matching}.
They defined $\gamma$-matchings in a very similar way. Their definition requires a time edge to be present~$\gamma$~consecutive time slots to be eligible for a temporal matching.
There is an easy reduction from their model to ours: For every sequence of~$\gamma$~consecutive time edges starting at time slot~$t$, we introduce \emph{only one} time edge at time slot~$t$, and set~$\Delta$ to~$\gamma$. 
This already implies that \textsc{Temporal Matching} is NP-complete~\cite[Theorem~1]{baste2018temporal} and that our algorithmic results also hold for \textsc{$\gamma$-Matching}.
We are not aware of an equally easy reduction in the reverse direction. 

In addition, it is easy to check that the algorithmic results of Baste~et~al.~\cite{baste2018temporal} also carry over to our model.
Hence, there is a 2-approximation algorithm for \textsc{Maximum Temporal Matching}%
~\cite[Corollary~1]{baste2018temporal} and \textsc{Temporal Matching} admits a polynomial kernel when parameterized by $k+\Delta$~\cite[Theorem~2]{baste2018temporal}. Some of our hardness results can also easily be transferred to \textsc{$\gamma$-Matching}.
Whenever this is the case, we will indicate this.%

\section{Hardness Results}\label{sec:hardness}
In this section we present our computational hardness results. In \cref{sec:apx}
we show that the optimization problem \textsc{Maximum Temporal Matching} is
APX-hard and in \cref{sec:pathhardness} we show that \textsc{Temporal
Matching} is NP-hard even if the underlying graph is a path.

\subsection{APX-completeness of \textsc{Maximum Temporal Matching}}
\label{sec:apx}

In this subsection, we look at \textsc{Maximum Temporal Matching} where we want to maximize the cardinality of the temporal matching.
We prove that \textsc{Maximum Temporal Matching} is APX-complete even if $\Delta=2$ and~$T=3$.
For this we provide an \emph{L-reduction}%
~\cite{ausiello2012complexity} %
from the APX-complete \textsc{Maximum Independent Set} problem on cubic graphs~\cite{alimonti2000some} to \textsc{Maximum Temporal Matching}. 
Together with the constant-factor approximation algorithm that we present in
Section \ref{sec:approxAlg}, this implies APX-completeness for \textsc{Maximum
Temporal Matching}. The reduction also implies NP-completeness of \textsc{Temporal Matching}. Formally, we show the following result.

\begin{theorem}%
	\label{th:apx}
	\textsc{Temporal Matching} is NP-complete and \textsc{Maximum Temporal Matching} is APX-complete even if
	$\Delta=2$, $T=3$, and every edge of the underlying graph appears only once.
	Furthermore, for any $\delta \geq \frac{664}{665}$, there is no polynomial-time $\delta$-approximation algorithm for
	{\sc Maximum Temporal Matching}, unless P $=$ NP.
	\end{theorem}

It is easy to check that the construction uses only three time steps and every edge appears in exactly one time step.
\begin{construction}\label{constr:lreduction}
Let $G = (V,E)$ be an $n$-vertex cubic graph. We construct in polynomial time a corresponding temporal graph $(H, \lambda)$
of lifetime three as follows. First, we find a proper 4-edge coloring $c: E \rightarrow \{1,2,3,4\}$ of $G$.
Such a coloring exists by Vizings's theorem and can be found in $O(|E|)$ time
\cite{schrijver1998bipartite}.
Now the underlying graph $H=(U,F)$ contains two vertices $v_0$ and $v_1$ for every vertex $v$ of $G$, 
and one vertex $w_e$ for every edge $e$ of $G$. 
The set~$F$ of the edges of $H$ contains $\{ v_0, v_1 \}$ for every $v \in V$, and for every edge $e=\{u,v\} \in E$ it
contains $\{w_e,u_{\alpha}\}, \{w_e,v_{\alpha}\}$, where $c(e) \equiv \alpha\ (\mathrm{mod}\ 2)$.
In the temporal graph $(H, \lambda)$ every edge of the underlying graph appears
in exactly one of the three time slots:
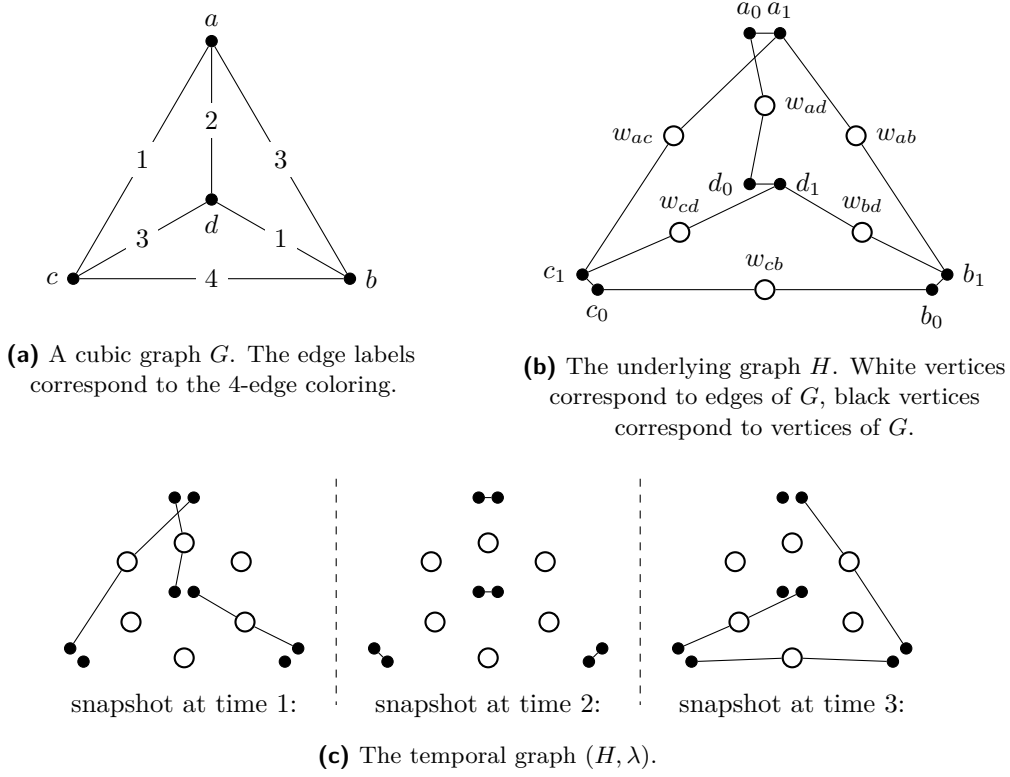
\begin{figure}[t]
	\centering
	\begin{subfigure}{.48\linewidth}\centering
			\begin{tikzpicture}[scale=0.7]
		\node[vertex,label=$a$] (a) at (0,3) {};
		\node[vertex,label=right:$b$] (b) at (2.6,-1.5) {};
		\node[vertex,label=left:$c$] (c) at (-2.6,-1.5) {};
		\node[vertex,label=below:$d$] (d) at (0,0) {};

		\coordinate [label=center:$ $] (G) at (0,-2.2);

		\draw (a) -- node[fill=white] {3} (b);
		\draw (a) -- node[fill=white] {1} (c);
		\draw (a) -- node[fill=white] {2} (d);
		\draw (b) -- node[fill=white] {4} (c);
		\draw (b) -- node[fill=white] {1} (d);
		\draw (c) -- node[fill=white] {3} (d);

	\end{tikzpicture}

    		\caption{\centering A cubic graph $G$. The edge labels correspond to the 4-edge coloring. \phantom{asdfasdfasdfasdfasdfasdf}}\label{fig:cubic}
    	\end{subfigure}
	~~~
    	\begin{subfigure}{.48\linewidth}\centering
				 	\begin{tikzpicture}[scale=0.8]
							\node[vertex,label=$a_0$] (a0) at (-0.25,3) {};
							\node[vertex,label=$a_1$] (a1) at (0.25,3) {};
		
							\node[vertex,label=below:$b_0$] (b0) at  (2.75,-1.25){};
							\node[vertex,label=right:$b_1$] (b1) at (3,-1) {};
		
							\node[vertex,label=below:$c_0$] (c0) at (-2.75,-1.25) {};
							\node[vertex,label=left:$c_1$] (c1) at (-3,-1) {};
		
							\node[vertex,label=left:$d_0$] (d0) at (-0.25,0.5) {};
							\node[vertex,label=right:$d_1$] (d1) at (0.25,0.5) {};

		\node[vertex2,label=$w_{cd}$] (wcd) at (-1.4,-0.3) {};
		\node[vertex2,label=$w_{bd}$] (wbd) at (1.6,-0.3) {};
		\node[vertex2,label=left:$w_{ac}$] (wac) at (-1.5,1.3) {};
		\node[vertex2,label=right:$w_{ab}$] (wab) at (1.5,1.3) {};
		\node[vertex2,label=right:$w_{ad}$] (wad) at (0,1.8) {};
		\node[vertex2,label=$w_{cb}$] (wcb) at (0,-1.25) {};
		
		\foreach \from/\to in {a0/a1,b0/b1,c0/c1,d0/d1,
						wac/a1,wac/c1,wad/a0,wad/d0,wab/a1,wab/b1,
						wbd/b1,wbd/d1,wcb/c0,wcb/b0,wcd/c1,wcd/d1}
		\draw (\from) -- (\to);

	\end{tikzpicture}
	
    		\caption{\centering The underlying graph $H$. White vertices correspond to edges of $G$, black vertices correspond to vertices of $G$.} %
  	 \end{subfigure}

   	\bigskip
   
   	\begin{subfigure}{.99\linewidth}\centering
				 	\begin{tikzpicture}[scale=0.5]
							\node at (0,-2.5) {snapshot at time $1$:};
							\node at (8,-2.5) {snapshot at time $2$:};
							\node at (16,-2.5) {snapshot at time $3$:};
							\draw[dashed] (4,3.5) to (4,-2.5);
							\draw[dashed] (12,3.5) to (12,-2.5);
							\begin{scope}
							\node[vertex] (a0) at (-0.25,3) {};
							\node[vertex] (a1) at (0.25,3) {};
		
							\node[vertex] (b0) at  (2.65,-1.35){};
							\node[vertex] (b1) at (3,-1) {};
		
							\node[vertex] (c0) at (-2.65,-1.35) {};
							\node[vertex] (c1) at (-3,-1) {};
		
							\node[vertex] (d0) at (-0.25,0.5) {};
							\node[vertex] (d1) at (0.25,0.5) {};

		\node[vertex2] (wcd) at (-1.4,-0.3) {};
		\node[vertex2] (wbd) at (1.6,-0.3) {};
		\node[vertex2] (wac) at (-1.5,1.3) {};
		\node[vertex2] (wab) at (1.5,1.3) {};
		\node[vertex2] (wad) at (0,1.8) {};
		\node[vertex2] (wcb) at (0,-1.25) {};
		
		\foreach \from/\to in {wac/a1,wac/c1,wad/a0,wad/d0,
						wbd/b1,wbd/d1}
		\draw (\from) -- (\to);

							\end{scope}
							\begin{scope}[xshift=8cm]
							\node[vertex] (a0) at (-0.25,3) {};
							\node[vertex] (a1) at (0.25,3) {};
		
							\node[vertex] (b0) at  (2.65,-1.35){};
							\node[vertex] (b1) at (3,-1) {};
		
							\node[vertex] (c0) at (-2.65,-1.35) {};
							\node[vertex] (c1) at (-3,-1) {};
		
							\node[vertex] (d0) at (-0.25,0.5) {};
							\node[vertex] (d1) at (0.25,0.5) {};

		\node[vertex2] (wcd) at (-1.4,-0.3) {};
		\node[vertex2] (wbd) at (1.6,-0.3) {};
		\node[vertex2] (wac) at (-1.5,1.3) {};
		\node[vertex2] (wab) at (1.5,1.3) {};
		\node[vertex2] (wad) at (0,1.8) {};
		\node[vertex2] (wcb) at (0,-1.25) {};
		
		\foreach \from/\to in {a0/a1,b0/b1,c0/c1,d0/d1}
		\draw (\from) -- (\to);

							\end{scope}
							\begin{scope}[xshift=16cm]
							\node[vertex] (a0) at (-0.25,3) {};
							\node[vertex] (a1) at (0.25,3) {};
		
							\node[vertex] (b0) at  (2.65,-1.35){};
							\node[vertex] (b1) at (3,-1) {};
		
							\node[vertex] (c0) at (-2.65,-1.35) {};
							\node[vertex] (c1) at (-3,-1) {};
		
							\node[vertex] (d0) at (-0.25,0.5) {};
							\node[vertex] (d1) at (0.25,0.5) {};

		\node[vertex2] (wcd) at (-1.4,-0.3) {};
		\node[vertex2] (wbd) at (1.6,-0.3) {};
		\node[vertex2] (wac) at (-1.5,1.3) {};
		\node[vertex2] (wab) at (1.5,1.3) {};
		\node[vertex2] (wad) at (0,1.8) {};
		\node[vertex2] (wcb) at (0,-1.25) {};

		\foreach \from/\to in {
						wab/a1,wab/b1,
						wcb/c0,wcb/b0,wcd/c1,wcd/d1}
		\draw (\from) -- (\to);

							\end{scope}
	\end{tikzpicture}

    		\caption{\centering  The temporal graph $(H,\lambda)$.}\label{fig:Temporal}
    	\end{subfigure}
    
    	\caption{Example of the reduction from \textsc{Maximum Independent Set} on cubic graphs to \textsc{Maximum Temporal Matching}.}
    	\label{fig:APXreduction}
\end{figure}
\begin{enumerate}
	\item $\lambda(\{w_e,u_{\alpha}\}) = \lambda(\{w_e,v_{\alpha}\}) = 1$, where $c(e) \equiv \alpha\ (\mathrm{mod}\ 2)$, 
	for every edge $e=\{u,v\} \in E$ such that $c(e) \in \{1,2\}$;
	\item $\lambda(\{v_0,v_1\}) = 2$ for every $v \in V$;
	\item $\lambda(\{w_e,u_{\alpha}\}) = \lambda(\{w_e,v_{\alpha}\}) = 3$, where $c(e) \equiv \alpha\ (\mathrm{mod}\ 2)$, 
	for every edge $e=\{u,v\} \in E$ such that $c(e) \in \{3,4\}$.
\end{enumerate}
\hfill$\lhd$
\end{construction}
\cref{constr:lreduction} is illustrated in \cref{fig:APXreduction}. %
We first show that if we find a 2-temporal matching in the constructed graph $(H,\lambda)$, then we can assume w.l.o.g.\ that if $\{u,v\}\in E$, then the temporal matching contains at most one of the time edges $(\{u_0,u_1\},2)$ and $(\{v_0,v_1\},2)$.
This will allow us to construct an independent set for the original graph $G$ from the temporal matching. 
\begin{lemma}%
\label{lem:canonicalMatching}
	Let $G=(V,E)$ be a cubic graph and let $(H,\lambda)$ be the temporal graph obtained by applying \cref{constr:lreduction} to $G$. Let $M$ be a 2-temporal matching of $(H,\lambda)$. 
	Then there exists a 2-temporal matching $M'$ of $(H,\lambda)$ such that $|M'| = |M|$, and for every edge $e=\{u,v\} \in E$
	the matching~$M'$ contains at most one of the time edges $(\{u_0,u_1\},2)$ and $(\{v_0,v_1\},2)$.
	Moreover, $M'$ can be constructed from $M$ in polynomial time.
\end{lemma}
\begin{proof}
	We prove the first part of the lemma by induction on the number of edges $\{u',v'\} \in E$ such that $M$ contains both
	$(\{u_0',u_1'\},2)$ and $(\{v_0',v_1'\},2)$. Let us denote this number by $k$.
	For $k\le 1$ the statement is trivial. Let $k > 1$, and let $e=\{u,v\} \in E$ be an edge
	such that both $(\{u_0,u_1\},2)$ and~$(\{v_0,v_1\},2)$ are in $M$. Without loss of generality we assume that $c(e) = 1$. 
	Since the lifetime of $(H,\lambda)$ is three and $(\{u_0,u_1\},2) \in M$, no time edge in $M$ other than $(\{u_0,u_1\},2)$ 
	is incident with~$u_0$ or~$u_1$. Similarly, no time edge in $M$ besides $(\{v_0,v_1\},2)$ is incident with~$v_0$ or~$v_1$.
	In particular, $(\{w_e,u_1\},1), (\{w_e,v_1\},1) \notin M$. Hence, $M''$ obtained from~$M$ by replacing $(\{u_0,u_1\},2)$ 
	with $(\{w_e,u_1\},1)$
	is a 2-temporal matching of $(H,\lambda)$ with~$|M''| = |M|$, and the number of edges $\{u',v'\} \in E$ such that
	$M''$ contains both $(\{u_0',u_1'\},2)$ and $(\{v_0',v_1'\},2)$ is~$k-1$. Hence, by the induction hypothesis,
	there exists a desired 2-temporal matching $M'$.

	Clearly, the above arguments can be turned into a polynomial-time algorithm
	that transforms $M$ into~$M'$ by iteratively finding edges $\{u',v'\} \in E$ such that both $(\{u_0',u_1'\},2)$ and $(\{v_0',v_1'\},2)$
	are in the current temporal matching and replacing one of the time edges by an appropriate incident time edge.
\end{proof}%

Next, we formally show how to obtain an independent set of $G$ from a 2-temporal matching of the constructed graph $(H,\lambda)$.
\begin{lemma}%
\label{lem:fromMatchingToIS}
	Let $G=(V,E)$ be a cubic graph and let $(H,\lambda)$ be the temporal graph obtained by applying \cref{constr:lreduction} to $G$. Let $M$ be a 2-temporal matching of $(H,\lambda)$. Then $G$ contains an independent set~$S$ of size
	at least $|M| - \frac{3n}{2}$. Moreover, $S$ can be computed from $M$ in polynomial time.
\end{lemma}
\begin{proof}
	First, by Lemma \ref{lem:canonicalMatching}, we can assume that for every $\{u,v\} \in E$
	the temporal matching~$M$ contains at most one of the time edges $(\{u_0,u_1\},2)$ and $(\{v_0,v_1\},2)$.
	Now we compute in polynomial time $S := \{ v~|~(\{v_0,v_1\},2) \in M \}$. 
	The above assumption implies that $S$ is an independent set.
	
	Furthermore, notice that for every edge $e \in E$ the underlying graph $H$ contains exactly 
	two edges incident with $w_e$ and both of them appear in the same time slot. 
	Hence $M$ can contain at most one time edge incident with $w_e$, and therefore 
	$|M| \leq |S| + |E| = |S| + \frac{3n}{2}$, which completes the proof.
\end{proof}%

Now we investigate how the size of a temporal matching in the constructed graph relates to the size the corresponding independent set in the original graph.
\begin{lemma}%
\label{lem:nuAlpha}
	Let $G=(V,E)$ be a cubic graph and let $(H,\lambda)$ be the temporal graph obtained by applying \cref{constr:lreduction} to $G$. 
	Let $\MaxTempMatchingSize_2$ be the size of a maximum 2-temporal matching in $(H,\lambda)$, and let $\alpha$ be 
	the size of a maximum independent set in~$G$. Then $\MaxTempMatchingSize_2 = \alpha + \frac{3n}{2}$.
\end{lemma}
\begin{proof}
	We start by proving $\MaxTempMatchingSize_2 \leq \alpha + \frac{3n}{2}$. Let $M$ be a maximum 2-temporal matching of~$(H,\lambda)$.
	By Lemma \ref{lem:fromMatchingToIS} there exists an independent set $S$ in~$G$ of size at least $|S| \geq |M| - \frac{3n}{2}$.
	Hence we have $\MaxTempMatchingSize_2 = |M| \leq |S| + \frac{3n}{2} \leq \alpha + \frac{3n}{2}$.
	
	To prove the converse inequality, we consider a maximum independent set $S$ in $G$, and show how to 
	construct a 2-temporal matching~$M$ of~$(H,\lambda)$ of size at least $|S| + \frac{3n}{2}$.
	First, for every $v \in S$ we include $(\{v_0,v_1\}, 2)$ in $M$. 
	Second, for every edge $e = \{u,v\} \in E$ we add one more time edge to $M$ as
	follows.
	Since~$S$ is independent, at least one of $u$ and $v$ is not in~$S$, say~$u$. 
	Then we add to $M$
	\begin{enumerate}
		\item $(\{w_eu_1\},1)$ if $c(e) = 1$,
		\item $(\{w_eu_0\},1)$ if $c(e) = 2$,
		\item $(\{w_eu_1\},3)$ if $c(e) = 3$, and
		\item $(\{w_eu_0\},3)$ if $c(e) = 4$.
	\end{enumerate}
	By construction we have $|M| = |S| + \frac{3n}{2}$. Now we show that $M$ is a 2-temporal matching.
	For any two distinct vertices $u$ and $v$ in $S$ the edges $\{u_0,u_1\}$ and $\{v_0,v_1\}$ are not adjacent in~$H$,
	therefore the time edges $(\{u_0,u_1\},2)$ and $(\{v_0,v_1\},2)$ are not in conflict.
	Furthermore, for any pair of adjacent edges $\{w_e,u_{\alpha}\}, \{u_0,u_1\}$ in $H$ the corresponding time edges 
	are not in conflict in $M$, as, by construction, at most one of them is in $M$. 
	For the same reason, for every edge $e=\{u,v\} \in E$ the time edges corresponding to $\{w_e,u_{\alpha}\}$ and 
	$\{w_e,v_{\alpha}\}$, where $c(e) \equiv \alpha\ (\mathrm{mod}\ 2)$, are not in conflict in $M$.
	It remains to show that the time edges $(\{w_e,u_\alpha\},i)$ and $(\{w_{e'},u_\alpha\},j)$ corresponding to the adjacent edges 
	$\{w_e,u_\alpha\}$ and $\{w_{e'},u_\alpha \}$ in $H$ are not in conflict in~$M$.
	Suppose to the contrary that the time edges are in conflict. Then both of them are in $M$ and $|i-j| \leq 1$. Since
	by definition~$i,j \in \{1,3\}$, we conclude that~$i=j$, i.e.,\ the time edges
	appear in the same time slot.
	Notice that~$e$ and~$e'$ share vertex~$u$, and hence~$c(e) \neq c(e')$. Hence, since 
	$c(e) \equiv \alpha\ (\mathrm{mod}\ 2)$ and~$c(e') \equiv \alpha\ (\mathrm{mod}\ 2)$, we conclude that either
	$\{ c(e), c(e') \} = \{ 1, 3 \}$, or $\{ c(e), c(e') \} = \{ 2, 4 \}$, but, by construction, this contradicts the assumption that $i = j$.
	This completes the proof that $M$ is a 2-temporal matching, and therefore
	we have $\MaxTempMatchingSize_2 \geq |M| = |S| + \frac{3n}{2} = \alpha+ \frac{3n}{2}$.
\end{proof}%

Lastly, we formally show that \cref{constr:lreduction} together with the prodecure described in \cref{lem:fromMatchingToIS} to obtain an independent set from a temporal matching is actually an L-reduction.
\begin{lemma}\label{lem:lred}
	\cref{constr:lreduction} together with the procedure described by \cref{lem:fromMatchingToIS} constitute an L-reduction.
\end{lemma}
\begin{proof}
	Recall the definition of an L-reduction~\cite{ausiello2012complexity}.
	Let $A$ and $B$ be two maximization problems and let $s_A$ and $s_B$ be their respective cost functions.
	By definition, a pair of functions $f$ and~$g$ is an L-reduction if all of the following conditions are met:
	\begin{enumerate}
		\item[(1)] functions $f$ and $g$ are computable in polynomial time;
		\item[(2)] if $I$ is an instance  of problem $A$, then $f(I)$ is an instance of problem $B$;
		\item[(3)] if $M$ is a feasible solution to $f(I)$, then $g(M)$ is a feasible solution to $I$;
		\item[(4)] there exists a positive constant $\beta$ such that $OPT_B(f(I)) \leq \beta \cdot OPT_A(I)$; and
		\item[(5)] there exists a positive constant $\gamma$ such that, for every feasible solution $M$ to~$f(I)$, it holds that 
		$OPT_A(I) - c_A(g(M)) \leq \gamma \cdot (OPT_B(f(I)) - c_B(M))$.
	\end{enumerate}
	
	In our case \textsc{Maximum Independent Set} in cubic graphs corresponds to problem~$A$ and
	\textsc{Maximum Temporal Matching} corresponds to problem~$B$. The reduction mapping a cubic graph $G$ to
	a temporal graph $(H,\lambda)$ described in \cref{constr:lreduction} corresponds to function~$f$. Clearly,
	the reduction is computable in polynomial time. The polynomial-time procedure guaranteed by 
	\cref{lem:fromMatchingToIS} corresponds to function $g$. It remains to show that conditions (4) and (5)
	in the definition of an L-reduction are met.
	
	By \cref{lem:nuAlpha} we know that 
	$\MaxTempMatchingSize_2(H,\lambda) = \alpha(G) + \frac{3n}{2} = \alpha(G) + \frac{6n}{4} \leq 7\alpha(G)$, where the latter inequality 
	follows from the fact that the independence number of an $n$-vertex cubic graph is at least $\frac{n}{4}$. %
	Hence, condition~(4) holds with parameter $\beta = 7$.
	
	Let now $M$ be a 2-temporal matching of $(H,\lambda)$, and let $S$ be an independent set in $G$ guaranteed by
	\cref{lem:fromMatchingToIS}, then
	$$
		\alpha(G) - |S| = \MaxTempMatchingSize_2(H,\lambda) - \frac{3n}{2} - |S| \leq \MaxTempMatchingSize_2(H,\lambda) - \frac{3n}{2} - |M| + \frac{3n}{2} = 
		\MaxTempMatchingSize_2(H,\lambda) - |M|,
	$$
	where the first equality follows from \cref{lem:nuAlpha} and the inequality follows from \cref{lem:fromMatchingToIS}.
	Thus, condition~(5) holds with parameter $\gamma = 1$.
\end{proof}
We now show that our reduction together with a polynomial-time $\delta$-approximation algorithm 
$\mathcal{A}$ for {\sc Maximum Temporal Matching}, where $\delta \geq \frac{664}{665}$, imply a polynomial-time 
$\frac{94}{95}$-approximation algorithm for {\sc Maximum Independent Set} in cubic graphs.
The result will then follow from the fact that it is NP-hard to approximate {\sc Maximum Independent Set} in cubic graphs 
to within factor of $\frac{94}{95}$ \cite{chlebik2006complexity}. 

Let $G$ be a cubic graph and $(H,\lambda)$ be the corresponding temporal graph from the reduction. 
Let also $M$ be a $2$-temporal matching found by algorithm $\mathcal{A}$, 
and let $S$ be the independent set in $G$ corresponding to $M$.
Since $\mathcal{A}$ is a $\delta$-approximation algorithm, we have $\frac{|M|}{\mu_2(H,\lambda)} \geq \delta$.
Furthermore, by \cref{lem:lred}, our reduction is an L-reduction with parameters $\beta = 7$ and $\gamma = 1$, that is, 
$\mu_2(H,\lambda) \leq 7 \alpha(G)$ and $\alpha(G) - |S| \leq \mu_2(H,\lambda) - |M|$.
Hence, we have
$$
	\alpha(G) - |S| \leq \mu_2(H,\lambda) - |M| \leq \mu_2(H,\lambda) \cdot (1-\delta) \leq 7 \alpha(G) \cdot (1-\delta),
$$
which together with $\delta \geq \frac{664}{665}$ imply $\frac{|S|}{\alpha(G)} \geq 7\delta - 6 \geq \frac{94}{95}$, as required.%

The Exponential Time Hypothesis (ETH) implies (together with the Sparsification Lemma) that there is no $2^{o(\text{\#variables }+\text{ \#clauses})}$-time algorithm for \textsc{3SAT}~\cite{ImpagliazzoP01,ImpagliazzoPZ01}. 
For \textsc{Temporal Matching} we can observe the following.
\begin{observation}
\textsc{Temporal Matching} does not admit a $2^{o(k)}\cdot |\mathcal{G}|^{f(T)}$-time algorithm for any function $f$, unless the Exponential Time Hypothesis fails.
\end{observation}
When investigating the original reduction from \textsc{3SAT} to \textsc{Vertex Cover} on cubic graphs~\cite{garey1974some}, it is easy to verify that the size of the constructed instance is linear in the size of the 3SAT formula. 
Hence, it follows that there is no $2^{o(|V|)}$-time algorithm for \textsc{Vertex Cover} on cubic graphs unless the ETH fails. 
It follows that there is no $2^{o(|V|)}$-time algorithm for \textsc{Independent Set} on cubic graphs unless the ETH fails. 
If we treat the reduction presented in \cref{constr:lreduction} as a polynomial-time many-one reduction, then we set the solution size for the \textsc{Temporal Matching} instance to the solution size of the \textsc{Independent Set} instance plus $3/2$ times the number of vertices in the \textsc{Independent Set} instance (see \cref{lem:fromMatchingToIS} and \cref{lem:nuAlpha}). 
It follows that the existence of a $2^{o(k)}\cdot |\mathcal{G}|^{f(T)}$-time algorithm (for any function~$f$) for \textsc{Temporal Matching} implies a $2^{o(|V|)}$-time algorithm for \textsc{Independent Set} on cubic graphs (note that $T$ is constant in the reduction), which contradicts the ETH.

\begin{observation}%
		\label{obs:complete}
\textsc{Temporal Matching} is NP-complete, even if $\Delta=2$, $T=5$, and the underlying graph of the input temporal graph is complete.
\end{observation}
\begin{proof}[Proof Sketch]
We observe that \cref{constr:lreduction} can be modified in such a way that it produces a temporal graph that has a complete underlying graph.
Namely, we can add two additional snapshots to the construction, one edgeless snapshot at time slot four, and one snapshot that is a complete graph at time slot five. This has the consequence that the size of the matching increases by exactly $\lfloor n/2 \rfloor$ and the underlying graph of the constructed temporal graph is a complete graph. Hence, we obtain \cref{obs:complete}.
\end{proof}

The importance of this observation is due to the following parameterized complexity implication.
Parameterizing \textsc{Temporal Matching} by structural graph parameters of the underlying graph that are constant on complete graphs cannot yield fixed-parameter tractability unless~P~$=$~NP, even if combined with the lifetime~$T$.
Note that many structural parameters fall into this category, such as domination number, distance to cluster graph, clique cover number, etc. 

\ifarxiv{}
\subparagraph*{Adapting \cref{constr:lreduction} to the Model of Baste~et~al.~\cite{baste2018temporal}.}\label{app:baste}
\fi{}
We remark that our reduction for \cref{th:apx} can easily be adapted to the model of Baste~et~al.~\cite{baste2018temporal}: recall that every edge of the underlying graph of the temporal graph constructed in the reduction (see \cref{constr:lreduction}) appears in exactly one time step.
Hence, for each of these time edges, we can add a second appearance exactly one
time step after the first appearance without creating any new matchable edges.
Of course, in order to do that for time edges appearing in the third time step,
we need another fourth time step. It follows that
\textsc{$\gamma$-Matching}~\cite{baste2018temporal} is NP-hard and its canonical
optimization version is APX-hard even if $\gamma =2$ and $T=4$, which
strengthens the hardness result by Baste~et~al.~\cite{baste2018temporal}.

\subsection{NP-completeness of \textsc{Temporal Matching} with underlying Paths}
\label{sec:pathhardness}

In this subsection we show NP-completeness of \textsc{Temporal Matching} even for a very restricted class of temporal graphs.

\begin{theorem}\label{thm:underlyingpath}
\textsc{Temporal Matching} is NP-complete even if $\Delta=2$ and the underlying graph of the input temporal graph is a path. %
\end{theorem}

We show this result by a reduction from \textsc{Independent Set} on connected cubic planar graphs, which is known to be NP-complete~\cite{garey1977rectilinear,GareyJ79}. More specifically, we show that \textsc{Independent Set} is NP-complete on the temporal line graphs of temporal graphs that have a path as underlying graph. Recall that by \cref{obs:TMequivMIS}, solving \textsc{Independent Set} on a temporal line graph is equivalent to solving \textsc{Temporal Matching} on the corresponding temporal graph.
We proceed as follows.
\begin{enumerate}
	\item We show that $2$-temporal line graphs of temporal graphs that have a path as underlying graph have a grid-like structure. More specifically, we show that they are induced subgraphs of so-called \emph{diagonal grid graphs} or \emph{king's graphs}%
~\cite{chang2013algorithmic,guo2018coloring}.
  \item We show that \textsc{Independent Set} is NP-complete on induced subgraphs of diagonal grid graphs which together with \cref{obs:TMequivMIS} yields \cref{thm:underlyingpath}.
  \begin{itemize}
  \item We exploit that cubic planar graphs are induced topological minors of grid graphs and extend this result by showing that they are also induced topological minors of diagonal grid graphs.
  \item We show how to modify the subdivision of a cubic planar graph that is an induced subgraph of a diagonal grid graph such that NP-hardness of finding independent sets of certain size is preserved.
  \end{itemize}
\end{enumerate}

\begin{figure}[t]
\centering
		\begin{subfigure}[c]{0.45\textwidth}
		\begin{center}
		\begin{tikzpicture}
\foreach \i in {0,...,5} {
      \node[vertex] (v1\i) at (1,\i) {};
    }
\foreach \i in {1,...,5} {
      \draw[edge] (1,5-\i+1) -- (1,5-\i) node [midway,right] {\footnotesize $e_{\i}, \lambda(e_{\i})=\{1,2,3,4,5\}$};

    }
    \phantom{\node (A) at (2,0) {\footnotesize $1$};}
\end{tikzpicture}
		\subcaption{\centering Temporal graph $\G=(P_6,\lambda)$ with $\lambda(e)=[5]$ for all~$e\in E(P_6)$. \phantom{asdfasdfasdsfdgsdfgsdfgfasdfasdf} \phantom{asdfasdfasdfasdfasdf}}
			\label{fig:temppath}
\end{center}
		\end{subfigure}
		\begin{subfigure}[c]{0.45\textwidth}
\begin{center}
\begin{tikzpicture}
\foreach \i in {1,...,5} {
\draw[gray, dotted] (\i,0.8) -- (\i,5.2);
\draw[gray, dotted] (0.8,\i) -- (5.2,\i);
}
\foreach \i in {1,...,5} {
\foreach \j in {1,...,5} {
      \node[vertex] (v\i\j) at (\i,\j) {};
    }
    \node (e\i) at (0.5,6-\i) {\footnotesize $e_{\i}$};
    \node (t\i) at (\i,0.5) {\footnotesize $\i$};
    }
\foreach \i in {1,...,4} {
\foreach \j in {1,...,4} {
      \draw[edge] (\i,\j) -- (\i+1,\j);
      \draw[edge] (\i,\j) -- (\i,\j+1);
      \draw[edge] (\i,\j) -- (\i+1,\j+1);
      \draw[edge] (\i+1,\j) -- (\i,\j+1);
    }
      \draw[edge] (5,\i) -- (5,\i+1);
      \draw[edge] (\i,5) -- (\i+1,5);
    }
    \phantom{\node (A) at (1,.5) {};}
    \node (A) at (1,5.5) {};
\end{tikzpicture}
		\subcaption{\centering $2$-Temporal line graph $L_2(\G)$. The horizontal dimension corresponds to time slots 1 to 5, the vertical dimension corresponds to the edges of $P_6$.}
			\label{fig:diaggrid}
\end{center}
\end{subfigure}
	\caption{
		A temporal line graph with a path as underlying graph where edges are always active and its $2$-temporal line graph.
		}
		\label{fig:dgg}
\end{figure}
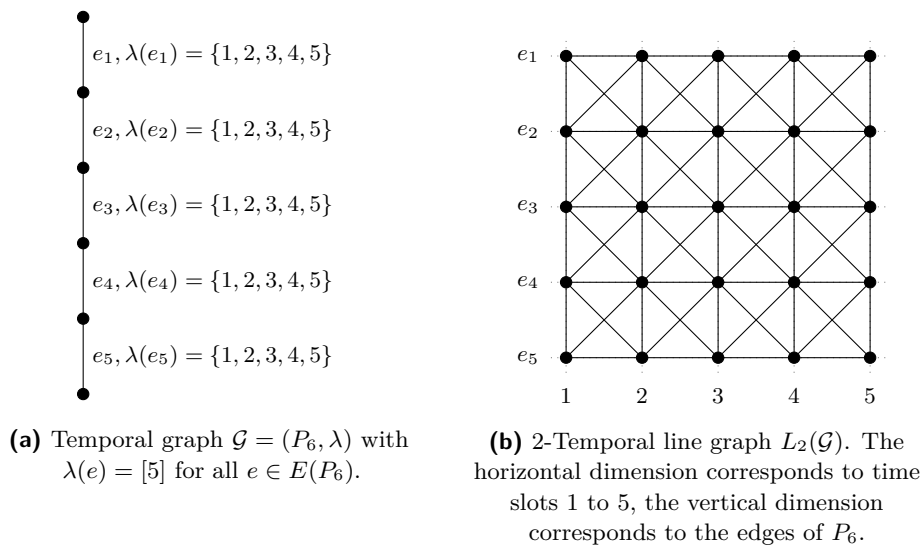

\begin{definition}[Diagonal Grid Graph~\cite{chang2013algorithmic,guo2018coloring}]\label{def:ddg}
A \emph{diagonal grid graph} $\dgg{n}{m}$ has a vertex $v_{i,j}$ for all $i\in[n]$ and $j\in[m]$ and there is an edge $\{v_{i,j},v_{i',j'}\}$ if and only if $|i-i'|^2+|j-j'|^2\le 2$.
\end{definition}

It is easy to check that for a temporal graph with a path as underlying graph and where each edge is active at every time step, the $2$-temporal line graph is a diagonal grid graph. %
\begin{observation}
Let $\G=(P_n,\lambda)$ with $\lambda(e)=[T]$ for all $e\in E(P_n)$, then $L_2(\G)=\dgg{n-1}{T}$.
\end{observation}
Further, it is easy to see that deactivating an edge at a certain point in time results in removing the corresponding vertex from the diagonal grid graph. See \cref{fig:dggsubgraph} for an example. Hence, we have that every induced subgraph of a diagonal grid graph is a $2$-temporal line graph.
\begin{corollary}\label{cor:grid}
Let $Z'$ be a connected \emph{induced} subgraph of $\dgg{n-1}{T}$. Then there is a $\lambda$ and an $n'\le n$ such that $Z'=L_2((P_{n'},\lambda))$.
\end{corollary}

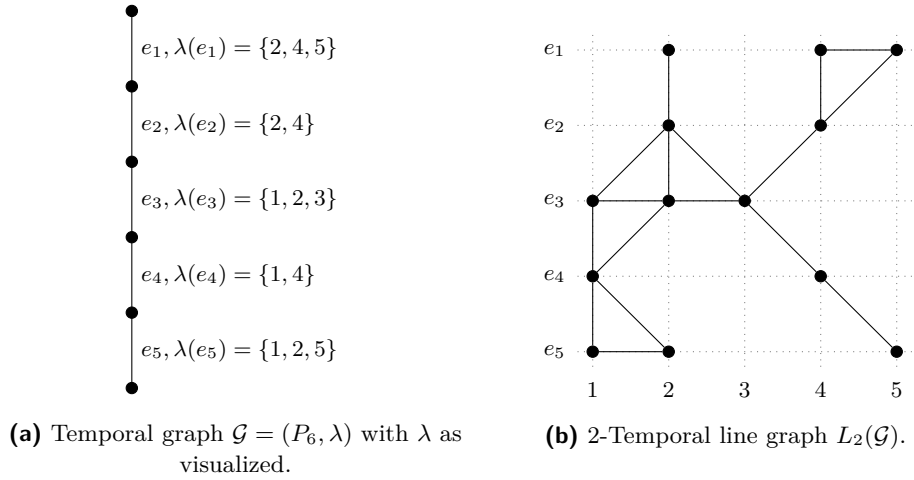
\begin{figure}[t]
\centering
		\begin{subfigure}[c]{0.45\textwidth}
		\begin{center}
		\begin{tikzpicture}
\foreach \i in {1,...,6} {
      \node[vertex] (v1\i) at (1,\i) {};
    }
      \draw[edge] (1,1) -- (1,2) node [midway,right] {\footnotesize $e_5, \lambda(e_5)=\{1,2,5\}$};
      \draw[edge] (1,2) -- (1,3) node [midway,right] {\footnotesize $e_4, \lambda(e_4)=\{1,4\}$};
      \draw[edge] (1,3) -- (1,4) node [midway,right] {\footnotesize $e_3, \lambda(e_3)=\{1,2,3\}$};
      \draw[edge] (1,4) -- (1,5) node [midway,right] {\footnotesize $e_2, \lambda(e_2)=\{2,4\}$};
      \draw[edge] (1,5) -- (1,6) node [midway,right] {\footnotesize $e_1, \lambda(e_1)=\{2,4,5\}$};

    \phantom{\node (A) at (2,1) {\footnotesize $1$};}
\end{tikzpicture}
		\subcaption{\centering Temporal graph $\G=(P_6,\lambda)$ with $\lambda$ as visualized.}
			\label{fig:temppath2}
\end{center}
		\end{subfigure}
		\begin{subfigure}[c]{0.45\textwidth}
\begin{center}
\begin{tikzpicture}[rotate=90,yscale=-1]
\foreach \i in {1,...,5} {
\draw[gray, dotted] (\i,0.8) -- (\i,5.2);
\draw[gray, dotted] (0.8,\i) -- (5.2,\i);
}
\foreach \i in {1,...,5} {
    \node (e\i) at (0.5,\i) {\footnotesize $\i$};
    \node (t\i) at (6-\i,0.5) {\footnotesize $e_{\i}$};
    }
    
      \node[vertex] (v11) at (1,1) {};
      \node[vertex] (v21) at (2,1) {};
      \node[vertex] (v31) at (3,1) {};
      
      \node[vertex] (v12) at (1,2) {};
      \node[vertex] (v32) at (3,2) {};
      \node[vertex] (v42) at (4,2) {};
      \node[vertex] (v52) at (5,2) {};
      
      \node[vertex] (v33) at (3,3) {};
      
      \node[vertex] (v24) at (2,4) {};
      \node[vertex] (v44) at (4,4) {};
      \node[vertex] (v54) at (5,4) {};
      
      \node[vertex] (v15) at (1,5) {};
      \node[vertex] (v55) at (5,5) {};

		\draw[edge] (v11) -- (v21);
		\draw[edge] (v21) -- (v31);
		
		\draw[edge] (v11) -- (v12);
		\draw[edge] (v21) -- (v12);
		\draw[edge] (v21) -- (v32);
		\draw[edge] (v31) -- (v32);
		\draw[edge] (v31) -- (v42);
		
		\draw[edge] (v32) -- (v42);
		\draw[edge] (v52) -- (v42);
		
		\draw[edge] (v32) -- (v33);
		\draw[edge] (v42) -- (v33);
		
		\draw[edge] (v24) -- (v33);
		\draw[edge] (v44) -- (v33);
		
		\draw[edge] (v44) -- (v54);
		
		\draw[edge] (v24) -- (v15);
		\draw[edge] (v44) -- (v55);
		\draw[edge] (v54) -- (v55);
		
		    \phantom{\node (A) at (.5,2) {};}
    \node (A) at (5.5,2) {};
\end{tikzpicture}
		\subcaption{\centering $2$-Temporal line graph $L_2(\G)$. \phantom{asdfasdfasdfasdfasdf}}
			\label{fig:inddiaggrid}
\end{center}
\end{subfigure}
	\caption{
		A temporal line graph with a path as underlying graph where edges are \emph{not} always active and its $2$-temporal line graph.
		}
		\label{fig:dggsubgraph}
\end{figure}

Having these results at hand, it suffices to show that \textsc{Independent Set} is NP-complete on induced subgraphs of diagonal grid graphs. By \cref{obs:TMequivMIS}, this directly implies that \textsc{Temporal Matching} is NP-complete on temporal graphs that have a path as underlying graph.
Hence, in the remainder of this section, we show the following result.
\begin{theorem}%
		\label{prop:nphgridis}
\textsc{Independent Set} on induced subgraphs of diagonal grid graphs is NP-complete. %
\end{theorem}
This result may be of independent interest and strengthens a result by Clark et~al.~\cite{clark1990unit}, who showed that \textsc{Independent Set} is NP-complete on unit disk graphs. It is easy to see from \cref{def:ddg} that diagonal grid graphs and their induced subgraphs are a (proper) subclass of unit disk graphs.

In the following, we give the main ideas of how we prove \cref{prop:nphgridis}. 
The first building block for the reduction is the fact that we can embed cubic planar graphs into a grid~\cite{valiant1981universality}.
More specifically, a cubic planar graph admits a planar embedding in such a way that the vertices are mapped to points of a grid and the edges are drawn along the grid lines. Moreover, such an embedding can be computed in polynomial time and the size of the grid is polynomially bounded in the size of the planar graph.

Note that if we replace the edges of the original planar graph by paths of appropriate length, then the embedding in the grid is actually a subgraph of the grid. Furthermore, if we scale the embedding by a factor of two, i.e.\ subdivide every edge once, then the embedding is also guaranteed to be an \emph{induced} subgraph of the grid. In other words, we argue that every cubic planar graph is an induced topological minor of a polynomially large grid graph. We then show how to modify the embedding in a way that insures that the resulting graph is also an induced topological minor of an polynomially large \emph{diagonal} grid graph. The last step is to further modify the embedding such that it can be obtained from the original graph by subdividing each edge an even number of times, this ensures that NP-hardness of \textsc{Independent Set} is preserved~\cite{poljak1974note}.

It is easy to check that \cref{prop:nphgridis}, \cref{obs:TMequivMIS}, and \cref{cor:grid} together imply \cref{thm:underlyingpath}.
\cref{thm:underlyingpath} also has some interesting implications from the point of view of parameterized complexity:
Parameterizing \textsc{Temporal Matching} by structural graph parameters of the underlying graph that are constant on a path cannot yield fixed-parameter tractability unless~P~$=$~NP, even if combined with~$\Delta$. %
Note that a large number of popular structural parameters fall into this category, such as maximum degree, treewidth, pathwidth, feedback vertex number, etc.

\ifarxiv{}
\subparagraph*{Proof of \cref{prop:nphgridis}.}\label{app:pathproof}
\fi{}
We prove \cref{prop:nphgridis} in several steps. We first use that a cubic planar graph admits a planar embedding in such a way that the vertices are mapped to points of a grid and the edges are drawn along the grid lines. Moreover, such an embedding can be computed in polynomial time and the size of the grid is polynomially bounded in the size of the planar graph. Furthermore, if we scale the embedding by a factor of two, i.e.\ subdivide every edge once, then the embedding is also guaranteed to be an \emph{induced} subgraph of the grid. In other words, we argue that every cubic planar graph is an induced topological minor of an polynomially large grid graph.

\begin{proposition}[Special case of Theorem~2 from Valiant {\cite{valiant1981universality}}]\label{cor:embedgrid}
Let $G=(V,E)$ be a cubic planar graph. Then~$G$ is an induced topological minor of $Z_{n,m}$ for some $n,m$ with $n\cdot m\in O(|V|^2)$ and the corresponding subdivision of $G$ can be computed in polynomial time.
\end{proposition}

We discuss next how to replace the edges of a cubic planar graph by paths of appropriate length such that it is an induced subgraph of a diagonal grid graph. In other words, we show that every cubic planar graph is an induced topological minor of a polynomially large diagonal grid graph.

\begin{lemma}%
\label{lem:embed}
Let $G=(V,E)$ be a cubic planar graph. Then $G$ is an induced topological minor of $\dgg{n}{m}$ for some $n,m$ with $n\cdot m\in O(|V|^2)$ and the corresponding subdivision of $G$ can be computed in polynomial time.
\end{lemma}
\begin{proof}
Let $G=(V,E)$ be a cubic planar graph. By \cref{cor:embedgrid} we know that there are integers $n, m$ with $n\cdot m\in O(|V|^2)$ such that~$G=(V,E)$ is an induced topological minor of~$Z_{n,m}$. 
Let $G'=(V',E')$ with $V'\subseteq \mathbb{N}\times\mathbb{N}$ be the corresponding subdivision of~$G$ that is an induced subgraph of~$Z_{n,m}$, i.e.~$Z_{n,m}[V']=G'$. Furthermore, for each vertex $v\in V$ of~$G$, let~$v'\in V'$ denote the corresponding vertex in the subdivision $G'$. 

Let $G''=(V'', E'')$ be the graph resulting from subdividing each edge in $G'$ eleven additional times and shift the graph three units away from the boundary of $Z_{n,m}$ in both dimensions. 
Intuitively, this is necessary to ensure that all paths in the grid are sufficiently far away from each other, which is also important in a later modification.

More formally, for each vertex $(i,j)\in V'$ create a vertex $(12i+3, 12j+3)\in V''$. 
For each edge $\{(i,j),(i,j+1)\}\in E'$ create eleven additional vertices, one for each grid point on the line between $(12i+3, 12j+3)$ and $(12i+3, 12j+15)$.
We connect these vertices by edges such that we get an induced path on the new vertices together with $(12i+3, 12j+3)$ and $(12i+3, 12j+15)$ that follows the grid line they lie on.
For each edge $\{(i,j),(i+1,j)\}\in E'$ we make an analogous modification to $G''$.
Furthermore, for each vertex $v\in V$ of~$G$, let $v''\in V''$ denote the corresponding vertex in the subdivision~$G''$. 
It is clear that $G''$ is an induced subgraph of $Z_{12n+6, 12m+6}$. We now show how to further modify~$G''$ such that it is an induced subgraph of the diagonal grid graph $\dgg{12n+6}{12m+6}$.

For each vertex $v\in V$ let $v''=(i,j)\in V''$, we check the following.
 \begin{enumerate}
   \item If $\deg_{G''}((i,j))=2$ and $\{(i,j),(i,j+1)\},\{(i,j),(i+1,j)\},\{(i,j),(i+2,j)\}\in E''$, then we delete $(i+1,j)$ from $V''$ and all its incident edges from $E''$. We add vertex $(i+1,j-1)$ to $V''$ and add edges $\{(i,j),(i+1,j-1)\}$ and $\{(i+1,j-1),(i+2,j)\}$ to~$E''$. This modification is illustrated in \cref{fig:embed2}. Rotated versions of this configuration are modified analogously.
   \item If $\deg_{G''}((i,j))=3$ and $\{(i,j),(i,j+1)\},\{(i,j),(i+1,j)\},\{(i,j),(i+2,j)\},$ $\{(i,j),(i-1,j)\},\{(i,j),(i-2,j)\}\in E''$, then we delete $(i+1,j)$ from $V''$ and all its incident edges from $E''$. We add vertex $(i+1,j-1)$ to $V''$ and add edges $\{(i,j),(i+1,j-1)\}$ and $\{(i+1,j-1),(i+2,j)\}$ to~$E''$. Furthermore, we we delete $(i-1,j)$ from $V''$ and all its incident edges from $E''$. We add vertex $(i-1,j-1)$ to $V''$ and add edges $\{(i,j),(i-1,j-1)\}$ and $\{(i-1,j-1),(i-2,j)\}$ to~$E''$. This modification is illustrated in \cref{fig:embed3}. Rotated versions of this configuration are modified analogously.
 \end{enumerate}
\begin{figure}[t]
\centering
		\begin{subfigure}[c]{0.3\textwidth}
\begin{center}
\begin{tikzpicture}[scale=.7]
\foreach \i in {1,...,4} {
\draw[gray, dotted] (\i,0.8) -- (\i,4.2);
}
\foreach \i in {1,...,4} {
\draw[gray, dotted] (0.8,\i) -- (4.2,\i);}

      \node[vertex] (v0) at (1,4) {};
      \node[vertex] (v1) at (1,3) {};
      \node[vertex2] (v2) at (1,2) {};
      \node[vertex] (v3) at (2,2) {};
      \node[vertex] (v4) at (3,2) {};
      \node[vertex] (v5) at (4,2) {};
      
      \draw[edge] (v0) -- (v1);
      \draw[edge] (v1) -- (v2);
      \draw[edge] (v2) -- (v3);
      \draw[edge] (v3) -- (v4);
      \draw[edge] (v4) -- (v5);

	  \draw[edge,dashed] (v1)-- (v3);
\foreach \i in {1,...,4} {
\draw[gray, dotted] (\i,0.8-5) -- (\i,4.2-5);
}
\foreach \i in {1,...,4} {
\draw[gray, dotted] (0.8,\i-5) -- (4.2,\i-5);}

      \node[vertex] (v02) at (1,4-5) {};
      \node[vertex] (v12) at (1,3-5) {};
      \node[vertex2] (v22) at (1,2-5) {};
      \node[vertex] (v32) at (2,1-5) {};
      \node[vertex] (v42) at (3,2-5) {};
      \node[vertex] (v52) at (4,2-5) {};
      
      \draw[edge] (v02) -- (v12);
      \draw[edge] (v12) -- (v22);
      \draw[edge] (v22) -- (v32);
      \draw[edge] (v32) -- (v42);
      \draw[edge] (v42) -- (v52);
\end{tikzpicture}
		\subcaption{\centering }
			\label{fig:embed2}
\end{center}
\end{subfigure}
		\begin{subfigure}[c]{0.4\textwidth}
\begin{center}
\begin{tikzpicture}[scale=.7]
\foreach \i in {1,...,7} {
\draw[gray, dotted] (\i,0.8) -- (\i,4.2);
}
\foreach \i in {1,...,4} {
\draw[gray, dotted] (0.8,\i) -- (7.2,\i);}

      \node[vertex] (v0) at (4,4) {};
      \node[vertex] (v1) at (4,3) {};
      \node[vertex2] (v2) at (4,2) {};
      \node[vertex] (v3) at (5,2) {};
      \node[vertex] (v4) at (6,2) {};
      \node[vertex] (v5) at (7,2) {};
      \node[vertex] (v6) at (3,2) {};
      \node[vertex] (v7) at (2,2) {};
      \node[vertex] (v8) at (1,2) {};
      
      \draw[edge] (v0) -- (v1);
      \draw[edge] (v1) -- (v2);
      \draw[edge] (v2) -- (v3);
      \draw[edge] (v3) -- (v4);
      \draw[edge] (v4) -- (v5);
      \draw[edge] (v2) -- (v6);
      \draw[edge] (v6) -- (v7);
      \draw[edge] (v7) -- (v8);
      
      \draw[edge, dashed] (v1) -- (v3);
      \draw[edge, dashed] (v1) -- (v6);

\foreach \i in {1,...,7} {
\draw[gray, dotted] (\i,0.8-5) -- (\i,4.2-5);
}
\foreach \i in {1,...,4} {
\draw[gray, dotted] (0.8,\i-5) -- (7.2,\i-5);}

      \node[vertex] (v02) at (4,4-5) {};
      \node[vertex] (v12) at (4,3-5) {};
      \node[vertex2] (v22) at (4,2-5) {};
      \node[vertex] (v32) at (5,1-5) {};
      \node[vertex] (v42) at (6,2-5) {};
      \node[vertex] (v52) at (7,2-5) {};
      \node[vertex] (v62) at (3,1-5) {};
      \node[vertex] (v72) at (2,2-5) {};
      \node[vertex] (v82) at (1,2-5) {};
      
      \draw[edge] (v02) -- (v12);
      \draw[edge] (v12) -- (v22);
      \draw[edge] (v22) -- (v32);
      \draw[edge] (v32) -- (v42);
      \draw[edge] (v42) -- (v52);
      \draw[edge] (v22) -- (v62);
      \draw[edge] (v62) -- (v72);
      \draw[edge] (v72) -- (v82);
\end{tikzpicture}
		\subcaption{\centering }
			\label{fig:embed3}
\end{center}
\end{subfigure}
\begin{subfigure}[c]{0.2\textwidth}
		\begin{center}
		\begin{tikzpicture}[scale=.7]
\foreach \i in {1,...,3} {
\draw[gray, dotted] (\i,0.8) -- (\i,3.2);
\draw[gray, dotted] (0.8,\i) -- (3.2,\i);
}

      \node[vertex] (v1) at (1,1) {};
      \node[vertex] (v2) at (1,2) {};
      \node[vertex] (v3) at (1,3) {};
      \node[vertex] (v4) at (2,3) {};
      \node[vertex] (v5) at (3,3) {};
      
      \draw[edge] (v1) -- (v2);
      \draw[edge] (v2) -- (v3);
      \draw[edge] (v3) -- (v4);
      \draw[edge] (v4) -- (v5);
      
      \draw[edge, dashed] (v2) -- (v4);

\foreach \i in {1,...,3} {
\draw[gray, dotted] (\i,0.8-4) -- (\i,3.2-4);
\draw[gray, dotted] (0.8,\i-4) -- (3.2,\i-4);
}

      \node[vertex] (v12) at (1,1-4) {};
      \node[vertex] (v22) at (1,2-4) {};
      \node[vertex] (v42) at (2,3-4) {};
      \node[vertex] (v52) at (3,3-4) {};
      
      \draw[edge] (v12) -- (v22);
      \draw[edge] (v22) -- (v42);
      \draw[edge] (v42) -- (v52);
        \phantom{
\foreach \i in {1,...,2} {
\draw[gray, dotted] (\i,0.8) -- (\i,4.2);}
\foreach \i in {1,...,2} {
\draw[gray, dotted] (\i,0.8-5) -- (\i,4.2-5);}
\foreach \i in {1,...,4} {
\draw[gray, dotted] (0.8,\i) -- (3.2,\i);}
\foreach \i in {1,...,4} {
\draw[gray, dotted] (0.8,\i-5) -- (3.2,\i-5);}
        }
\end{tikzpicture}
		\subcaption{\centering }
			\label{fig:embed1}
\end{center}
		\end{subfigure}
	\caption{
Illustration of the modifications described in the proof of \cref{lem:embed}. The situation before the moficication is depiced above, dashed edges show unwanted edges present in an induced subgraph of a diagonal grid graph. The situation after the modification is depicted below.
}
		\label{fig:embed}
\end{figure}
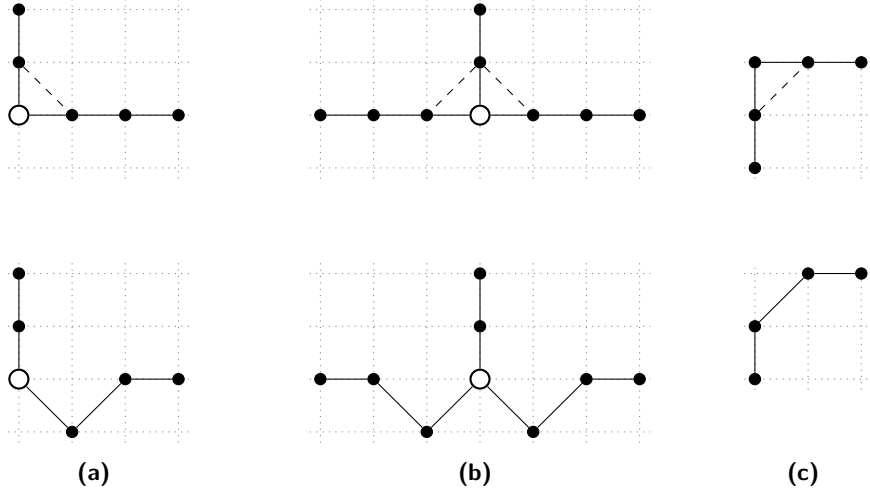
Lastly, whenever a path in $G''$ that corresponds to an edge in $G$ bends at a square angle, we remove the corner vertex and its incident edges and reconnect the path by a diagonal edge. 

More formally, let $(i,j-1),(i,j),(i+1,j)\in V''$ be adjacent vertices in a path in~$G''$ that corresponds to an edge in $G$, then we remove $(i,j)$ from $V''$ and all its incident edges and add the edge $\{(i,j-1),(i+1,j)\}$ to $E''$.
This modification is illustrated in \cref{fig:embed1}. Rotated versions of this configuration are modified analogously.

Now it is easy to see that $G''$ is an induced subgraph of $\dgg{12n+6}{12m+6}$. Furthermore, $G''$ can be computed in polynomial time.
\end{proof}%

Next we argue that we can always embed a cubic planar graph into a diagonal grid
graph in a way that preserves NP-hardness. This is based on the observation that
subdividing an edge of a graph two times increases the size of a maximum
independent set exactly by one.

\begin{observation}[Poljak \cite{poljak1974note}]\label{obs:is}
Let $G=(V,E)$ be a graph. Then for every $\{u,v\}\in E$, the graph $G'=(V\cup\{u',v'\}, (E\setminus\{\{u,v\}\})\cup\{\{u,u'\},\{u',v'\},\{v',v\}\})$ contains an independent set of size~$k+1$ if and only if $G$ contains an independent set of size $k$.
\end{observation}

From this observation it follows that if we can guarantee that for every cubic
planar graph there is a subdivision that subdivides every edge an even number of times and that is an induced subgraph of a diagonal grid graph of polynomial size, then we are done.

\begin{lemma}%
\label{lem:embed2}
Let $G=(V,E)$ be a cubic planar graph. Then there is a subdivision of $G$ that is an induced subgraph of $\dgg{n}{m}$ for some $n,m$ with $n\cdot m\in O(|V|^2)$ and where each edge of $G$ is subdivided an even number of times. Furthermore, the subdivision of $G$ can be computed in polynomial time.
\end{lemma}
\begin{proof}
Let $G=(V,E)$ be a cubic planar graph. By \cref{lem:embed} we know that there are some~$n, m$ with $n\cdot m\in O(|V|^2)$ such that~$G=(V,E)$ is an induced topological minor of~$\dgg{n}{m}$. 
Let $G'=(V',E')$ with $V'\subseteq \mathbb{N}\times\mathbb{N}$ be a subdivision of~$G$ constructed as described in the proof of \cref{lem:embed}.

Recall that every edge $e$ in $G$ is replaced by a path $P_e$ in $G'$. From
\cref{obs:is} it follows that if we can guarantee that all these paths have an
odd number of edges (and hence result from an even number of subdivisions), then $G'$ contains an independent set of size $k+\sum_{e\in E}\lfloor \frac{|E(P_e)|-1}{2}\rfloor$ if and only if $G$ contains an independent of size $k$.
In the following we show how to change the parity of the number of edges of a path~$P_e$ in $G'$ that corresponds to an edge $e$ in $G$.

The number of subdivisions performed in the construction that is described in
the proof of \cref{lem:embed} ensures that each path $P_e$ in $G'$ that corresponds to an edge $e$ in $G$ contains seven consecutive edges that are either all horizontal or all vertical.
Assume that $P_e$ contains an even number of edges and contains horizontal edges $\{(i,j),(i+1,j)\},\{(i+1,j),(i+2,j)\},\{(i+2,j),(i+3,j)\},\{(i+3,j),(i+4,j)\},\{(i+4,j),(i+5,j)\},\{(i+5,j),(i+6,j)\},\{(i+6,j),(i+7,j)\}$.
We remove vertices $(i+2,j), (i+3,j), (i+5,j)$ and all their incident edges. We add vertices $(i+2,j+1), (i+3,j+2),(i+4,j+1),(i+5,j-1)$ and edges $\{(i+1,j),(i+2,j+1)\},\{(i+2,j+1),(i+3,j+2)\},\{(i+3,j+2),(i+4,j+1)\},\{(i+4,j+1),(i+4,j)\},\{(i+4,j),(i+5,j-1)\},\{(i+5,j-1),(i+6,j)\}$. It is easy to check that this reconnects the path and increases the number of edges by one.
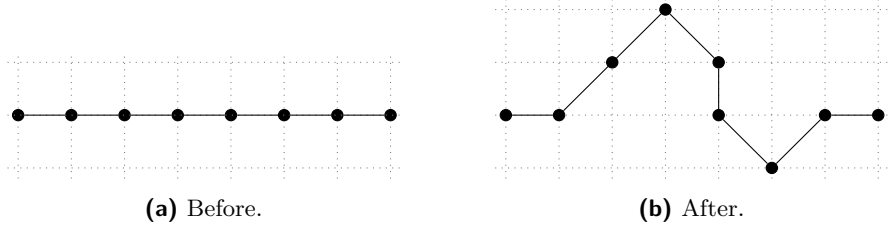
\begin{figure}[t]
\centering
		\begin{subfigure}[c]{0.45\textwidth}
		\begin{center}
		\begin{tikzpicture}[scale=.7]
\foreach \i in {1,...,8} {
\draw[gray, dotted] (\i,0.8) -- (\i,3.2);

      \node[vertex] (v\i) at (\i,2) {};
}
\foreach \i in {1,...,3} {
\draw[gray, dotted] (0.8,\i) -- (8.2,\i);
}
      
      \draw[edge] (v1) -- (v2);
      \draw[edge] (v2) -- (v3);
      \draw[edge] (v3) -- (v4);
      \draw[edge] (v4) -- (v5);
      \draw[edge] (v5) -- (v6);
      \draw[edge] (v6) -- (v7);
      \draw[edge] (v7) -- (v8);

\phantom{
\foreach \i in {1,...,8} {
\draw[gray, dotted] (\i,0.8) -- (\i,4.2);}
\foreach \i in {1,...,4} {
\draw[gray, dotted] (0.8,\i) -- (8.2,\i);
}
}
\end{tikzpicture}
		\subcaption{\centering Before.}
			\label{fig:parity1}
\end{center}
\end{subfigure}
		\begin{subfigure}[c]{0.45\textwidth}
		\begin{center}
		\begin{tikzpicture}[scale=.7]
\foreach \i in {1,...,8} {
\draw[gray, dotted] (\i,0.8) -- (\i,4.2);}
\foreach \i in {1,...,4} {
\draw[gray, dotted] (0.8,\i) -- (8.2,\i);
}

      \node[vertex] (v12) at (1,2) {};
      \node[vertex] (v22) at (2,2) {};
      \node[vertex] (v32) at (3,3) {};
      \node[vertex] (v42) at (4,4) {};
      \node[vertex] (v52) at (5,3) {};
      \node[vertex] (v62) at (5,2) {};
      \node[vertex] (v72) at (6,1) {};
      \node[vertex] (v82) at (7,2) {};
      \node[vertex] (v92) at (8,2) {};
      
      \draw[edge] (v12) -- (v22);
      \draw[edge] (v22) -- (v42);
      \draw[edge] (v42) -- (v52);
      \draw[edge] (v52) -- (v62);
      \draw[edge] (v62) -- (v72);
      \draw[edge] (v72) -- (v82);
      \draw[edge] (v82) -- (v92);
\end{tikzpicture}
		\subcaption{\centering After.}
			\label{fig:parity2}
\end{center}
\end{subfigure}
	\caption{
Illustration of the modification described in the proof of \cref{lem:embed2}.	It shows how to increase the length of an induced path of a diagonal grid graph by one.
}
		\label{fig:pathparity}
\end{figure}
This modification is illustrated in \cref{fig:pathparity}. The vertical version of this configuration is modified analogously.

Using this modification we can easily modify $G'$ in polynomial time in a way that all paths that correspond to edges of $G$ have an odd number of edges.
\end{proof}%

Now, \cref{prop:nphgridis} follows directly from \cref{lem:embed2} and \cref{obs:is}. 

\section{Algorithms}

\label{sec:algos}
In this section,
we present one approximation algorithm (\cref{sec:approxAlg}) and 
two exact fixed-parameter algorithms (\cref{subs:fptk} and \cref{subs:fptdeltanu})  for \textsc{Temporal Matching}.
First, in \cref{sec:approxAlg} we present an $\frac{\Delta}{2\Delta-1}$-approximation algorithm for \textsc{Maximum Temporal Matching}.
Second, in \cref{subs:fptk} we present an FPT-algorithm for the solution size parameter $k$.
Third, in \cref{subs:fptdeltanu} we present an FPT-algorithm for the parameter
combination time window size $\Delta$ and size $\nu$ of a maximum matching in
the underlying graph (which as a parameterization is incomparable to
parameterizing with the solution size $k$).

\subsection{Approximation of \textsc{Maximum Temporal Matching}}
\label{sec:approxAlg}

In this subsection, we present a $\frac{\Delta}{2\Delta-1}$-approximation
algorithm for \textsc{Maximum Temporal Matching}.
Note that for $\Delta=2$ this is a $\frac{2}{3}$-approximation, while for arbitrary constant $\Delta$ this is a $(\frac{1}{2}+\varepsilon)$-approximation,
where $\varepsilon = \frac{1}{2(2\Delta-1)}$ is a constant too.
Specifically, we show the following.
\begin{theorem}%
		\label{th:approxAlg}
	\textsc{Maximum Temporal Matching} admits an $O\left(Tm( \sqrt{n}+\Delta) \right)$-time 
	$\frac{\Delta}{2\Delta-1}$-approximation algorithm.
\end{theorem}

The main idea of our approximation algorithm is to compute maximum matchings for slices of size $\Delta$ of the input temporal graph 
that are sufficiently far apart from each other such that they do not interfere with each other, and hence are computable in polynomial time. 
Then we greedily fill up the gaps. 
We try out certain combinations of non-interfering slices of size $\Delta$ in a systematic way and 
then take the largest $\Delta$-matching that was found in this way. 
With some counting arguments we can show that this achieves the desired approximation ratio. 
In the following we describe and prove this claim formally.

We first introduce some additional notation and terminology. 
Recall that $\MaxTempMatchingSize_{\Delta}(\G)$ denotes the size of a maximum $\Delta$-temporal matching in $\G$.
Let~$\Delta$ and~$T$ be fixed natural numbers such that $\Delta \leq T$.
For every time slot $t\in [T-\Delta +1]$, we define the \emph{$\Delta$-window} $W_{t}$ as the interval $[t,t+\Delta -1]$ 
of length $\Delta$. We use this to formalize slices of size $\Delta$ of a temporal graph.
An interval of length at most $\Delta-1$ that either starts at slot 1, or ends at slot $T$ is called
a \textit{partial $\Delta$-window (with respect to lifetime $T$)}. 
For the sake of brevity, we write \textit{partial $\Delta$-window}, when the lifetime~$T$ is clear from the context.
The \emph{distance} between two disjoint intervals $[a_1,b_1]$ and $[a_2,b_2]$ with $b_1 < a_2$ is $a_2 - b_1 - 1$.

A \textit{$\Delta$-template (with respect to lifetime $T$)} is a maximal family $\mathcal{S}$ of $\Delta$-windows or 
partial $\Delta$-windows in the interval $[T]$ such that any two consecutive elements in $\mathcal{S}$ are at distance
exactly $\Delta-1$ from each other. 
Let $\mathS$ be a $\Delta$-template. A $\Delta$-temporal matching $M^\mathS$ in $\G = (G,\lambda)$ is called
a $\Delta$-temporal matching \emph{with respect to $\Delta$-template $\mathS$} if $M^\mathS$ has the maximum possible 
number of edges in every interval $W \in \mathS$, i.e.\ $\big|M^\mathS |_W\big| = \MaxTempMatchingSize_{\Delta}(\G|_W)$ for every $W \in \mathS$.

Now we are ready to present and analyze our $\frac{\Delta}{2\Delta-1}$-approximation algorithm, see \cref{alg:slidingTemplate}.
The idea of the algorithm is simple: for every $\Delta$-template $\mathS$ compute 
a $\Delta$-temporal matching $M^{\mathS}$ with respect to $\mathS$ and 
among all of the computed $\Delta$-temporal matchings return a matching of the maximum cardinality.
\begin{algorithm2e}[t]

	$M \gets \emptyset$.\;
	\ForEach(\label{line:approx-for}){$\Delta$-template $\mathS$}{
		Compute a $\Delta$-temporal matching $M^{\mathS}$ with respect to $\mathS$. \label{line:TemplateTM}\;
		\lIf{$|M^\mathS| > |M|$}{
			$M \gets M^\mathS$.
		}
	}
	\KwRet{$M$.}
	\caption{$\frac{\Delta}{2\Delta-1}$-Approximation Algorithm %
	(\cref{th:approxAlg}).}  
\label[algorithm]{alg:slidingTemplate}  
\end{algorithm2e}
The notions of $\Delta$-window, partial $\Delta$-window, and $\Delta$-template 
are illustrated in  \cref{fig:DeltaTemplate}.
A time slot $t$ is \textit{covered} by a $\Delta$-template $\mathS$ if $t$ belongs to an interval of $\mathS$.
\begin{figure}[t]
	\centering
	\includegraphics[width=\linewidth]{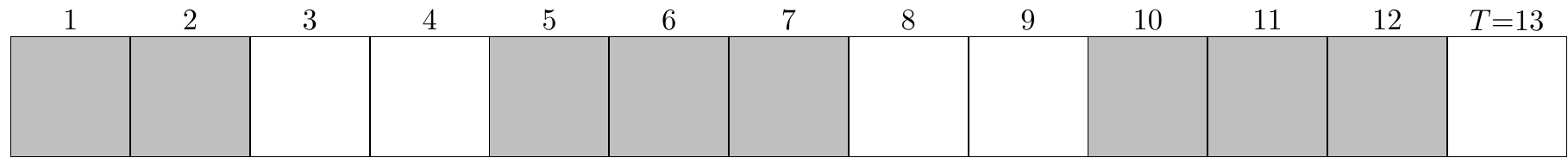}  
	\caption{The gray slots form the intervals of a $\Delta$-template, where $\Delta=3$. Interval $[1,2]$ is a partial $\Delta$-window.
	Intervals $[5,7]$ and $[10,12]$ are $\Delta$-windows.}
	\label{fig:DeltaTemplate}
\end{figure}
We show the following properties of $\Delta$-templates which we need to prove the approximation ratio of our algorithm.

\begin{lemma}
	\label{lem:DeltaTemplates}
	Let $\Delta$ and $T$ be natural numbers such that $\Delta \leq T$. Then
	\begin{enumerate}
		\item[(1)] there are exactly $2\Delta - 1$ different $\Delta$-templates with respect to lifetime $T$;
		\item[(2)] every time slot in $[T]$ is covered by exactly $\Delta$ different $\Delta$-templates.
	\end{enumerate} 
\end{lemma}
\begin{proof}
	To prove (1), we first observe that a $\Delta$-template $\mathS$ is uniquely determined by its leftmost interval.
	Indeed, by fixing the leftmost interval of $\mathS$, by definition, the subsequent intervals of $\mathS$ are located 
	in $[T]$ uniformly at distance exactly $\Delta-1$ from each other.
	Now, the maximality of $\mathS$ implies that the first interval in $\mathS$ is either a partial $\Delta$-window
	that starts at time slot 1 or a (possibly partial) $\Delta$-window that starts in one of the first $\Delta$ time slots of $[T]$. 
	Since there are $\Delta-1$ intervals of the first type and $\Delta$ intervals of the second type,
	we conclude that there are exactly $2\Delta - 1$ different $\Delta$-templates with respect to lifetime~$T$.
	
	To prove (2), we note that
	all $\Delta$-templates can be successively obtained from the $\Delta$-template $\mathS$ whose first interval
	is the single-slot partial $\Delta$-window $[1]$ by shifting by one time slot to the right all the intervals of 
	the current $\Delta$-template 
	(in each shift we augment the leftmost interval if it was a partial $\Delta$-window and truncate the rightmost
	interval if it covered the last time slot $T$).
	It is easy to see that every time slot will be covered in exactly~$\Delta$ of $2\Delta-1$ shifting iterations.
\end{proof}

By definition, for any two distinct intervals $W_1, W_2$ in a $\Delta$-template $\mathS$ and for any two
time slots $t_1 \in W_1$ and $t_2 \in W_2$ we have $|t_1-t_2| > \Delta$, which implies that no two time edges of $\G$ that appear in time slots of different intervals of $\mathS$ are in conflict.
This observation together with the fact that every interval in $\mathS$ is of length at most $\Delta$ imply that
a $\Delta$-temporal matching with respect to $\mathS$ can be computed in polynomial time by computing a
maximum $\Delta$-temporal matching in $\G|_{W}$ for every $W \in \mathS$ and then taking the union of these
matchings\footnote{The obtained $\Delta$-temporal matching can further be extended greedily to a maximal $\Delta$-temporal matching.}.
Since every $\Delta$-template has $O\left( \frac{T}{\Delta} \right)$ intervals and,
a maximum $\Delta$-temporal matchings in $\G|_{W}$, $W \in \mathS$ can be computed in $O(m(\sqrt{n} + \Delta))$ time,
which follows from \cref{obs:DeltaEqT},
we conclude that a $\Delta$-temporal matching with respect to $\mathS$ can be computed in 
$O\left(Tm \left(\frac{\sqrt{n}}{\Delta} + 1 \right)\right)$ time.

\begin{lemma}
	\label{lem:approxAlg}
	\cref{alg:slidingTemplate} is an $O\left(Tm( \sqrt{n}+\Delta) \right)$-time 
	$\frac{\Delta}{2\Delta-1}$-approximation algorithm for \textsc{Maximum Temporal Matching}.
\end{lemma}
\begin{proof}
	Let $\G = (G, \lambda)$ be an arbitrary temporal graph of lifetime $T$ and $\Delta$ be a natural number such that $\Delta \leq T$.
	Let also $M^*$ be a maximum $\Delta$-temporal matching of $\G$.
	
	We show that, given the instance $(\G, \Delta)$ of \textsc{Maximum Temporal Matching}, \cref{alg:slidingTemplate} produces in time 
	$O\left(Tm( \sqrt{n}+\Delta) \right)$
	a $\Delta$-temporal matching $M$ of size at least $\frac{\Delta}{2\Delta-1}|M^*|$, where $n$ and $m$ are the number
	of vertices and the number of edges in the underlying graph $G$, respectively.
	
	Clearly, the algorithm outputs a feasible solution as $M$ is a $\Delta$-temporal matching with respect to
	some $\Delta$-template. We show next that $M$ is the desired approximate solution.
	As in the pseudocode of \cref{alg:slidingTemplate}, for a $\Delta$-template $\mathcal{S}$ we denote by 
	$M^{\mathS}$ the $\Delta$-temporal matching with respect to $\mathS$ computed in Line \ref{line:TemplateTM} of
	 \cref{alg:slidingTemplate}. 
	Let $\mathfrak{S}$ be the family of all $\Delta$-templates with respect to lifetime $T$, and
	let $\mathS' \in \mathfrak{S}$ be a $\Delta$-template such that $M = M^{\mathS'}$.
	It follows from the algorithm that $|M^{\mathS'}| \geq |M^{S}|$ for every $S \in \mathfrak{S}$.
	By definition, for every $\mathS \in \mathfrak{S}$ and for every interval $W \in \mathS$ we have
	$
		\sum_{t \in W} |M^{\mathS}_t| \geq \sum_{t \in W} |M_t^*|,
	$
	where $M_t =  M \cap E_t$.
	Hence
	
	$$
		|M^{\mathS}| \geq \sum_{W \in \mathS} \sum_{t \in W} |M^{\mathS}_t| \geq \sum_{W \in \mathS} \sum_{t \in W} |M_t^*|.
	$$
	
	\noindent
	Using the above inequalities and \cref{lem:DeltaTemplates} we derive
	\begin{align*}
		(2\Delta-1) |M^{\mathS'}|
		&\geq
		\sum_{\mathS \in \mathfrak{S}}|M^{\mathS}|\\ &\geq
		\sum_{\mathS \in \mathfrak{S}}\sum_{W \in \mathS} \sum_{t \in W} |M^{\mathS}_t| \geq 
		\sum_{\mathS \in \mathfrak{S}}\sum_{W \in \mathS} \sum_{t \in W} |M_t^*| = 
		\Delta \sum_{t=1}^{T} |M_t^*| = \Delta |M^*|,
	\end{align*}
	which implies the $|M| = |M^{\mathS'}| \geq \frac{\Delta}{2\Delta-1}|M^*|$.
	
	Now we analyze the time complexity of the algorithm. By \cref{lem:DeltaTemplates} there are exactly $2\Delta-1$
	different $\Delta$-templates, and therefore the for-loop in Line \ref{line:approx-for} of \cref{alg:slidingTemplate} performs exactly $2\Delta-1$
	iterations. At every iteration the algorithm computes a $\Delta$-temporal matching with respect to
	a $\Delta$-template, which, as we discussed, can be done in $O\left(Tm \left(\frac{\sqrt{n}}{\Delta} + 1 \right)\right)$ time. 
	Altogether, the total time complexity is $O\left(Tm (\sqrt{n}+\Delta \right))$, as claimed.
\end{proof}

We remark that our analysis ignores the fact that the algorithm may add 
time edges from the gaps between the $\Delta$-windows 
defined by the template to the matching 
if they are not in conflict with any other edge in the matching. 
Hence, there is potential room for improvement.
However, %
our analysis of the approximation factor of \cref{alg:slidingTemplate} is tight for $\Delta = 2$.
Namely, there exists a temporal graph $\G$ (see \cref{fig:2/3Approximation}) such that 
on the instance $(\G,2)$ our algorithm (in the worst case) finds
a 2-temporal matching of size two, while the size of a maximum 2-temporal matching in $\G$ is three. 
In this example any improvement
of the algorithm that utilizes the gaps between the $\Delta$-windows would not lead to a better performance.

\begin{figure}[t]
	\centering
\begin{tikzpicture}
      \node[vertex] (v1) at (0,0) {};
      \node[vertex] (v2) at (3,0) {};
      \node[vertex] (v3) at (6,0) {};
      \node[vertex] (v4) at (9,0) {};
      \node[vertex] (v5) at (12,0) {};
      \draw[edge] (v1) -- node[above] {\footnotesize $e_{1}$} node[below] {\footnotesize $\lambda(e_{1})=\{2\}$} (v2);
      \draw[edge] (v2) -- node[above] {\footnotesize $e_{2}$} node[below] {\footnotesize $\lambda(e_{2})=\{1,3\}$} (v3);
      \draw[edge] (v3) -- node[above] {\footnotesize $e_{3}$} node[below] {\footnotesize $\lambda(e_{3})=\{1\}$} (v4);
      \draw[edge] (v4) -- node[above] {\footnotesize $e_{4}$} node[below] {\footnotesize $\lambda(e_{4})=\{2\}$} (v5);
\end{tikzpicture}
    	\caption{A temporal graph witnessing that the analysis of \cref{alg:slidingTemplate} is tight for $\Delta=2$.}
    	\label{fig:2/3Approximation}
\end{figure}
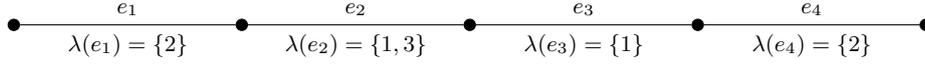

\subsection{Fixed-parameter tractability for the parameter solution size} \label{subs:fptk}

In this subsection, we provide a fixed-parameter algorithm 
for \textsc{Temporal Matching} parameterized by the solution size $k$.
More specifically, we provide a linear-time algorithm for a fixed solution size $k$.
Formally, the main result of this subsection is to show the following.
\begin{theorem}%
	\label{thm:fpt-for-k}
	\textsc{Temporal Matching} can be solved in $k^{O(k)}\cdot |\G|$ time.
\end{theorem}
We prove \cref{thm:fpt-for-k} in the remainder of this section.
Due to Baste~et~al.~\cite{baste2018temporal} it is already known that \textsc{Temporal Matching} is fixed-parameter tractable 
when parameterized by the solution size~$k$ \emph{and}~$\Delta$. %
In comparison to the algorithm of Baste~et~al.~\cite{baste2018temporal}
the %
running time of our algorithm is faster and independent of~$\Delta$,
hence improving their result from a parameterized classification standpoint.

The rough idea of our algorithm is the following. 
We develop a preprocessing procedure that reduces the number of time edges of the first $\Delta$-window. 
After applying this procedure, the number of time edges in the first $\Delta$-window is bounded in a function of the solution size parameter $k$.
Then, a search tree algorithm branches on which of the time edges from the first $\Delta$-window is in the solution and solve the remaining part recursively. %

Next, we describe the preprocessing procedure more precisely. 
Referring to kernelization algorithms, we call the result of this procedure \emph{kernel for the first $\Delta$-window}. 
If we count naively the number of $\Delta$-temporal matchings 
in the first $\Delta$-window of a temporal graph, 
then this number clearly depends on $\Delta$.
 This is too large for \cref{thm:fpt-for-k}.
A key observation to overcome this obstacle is that if we look at an edge appearance of a $\Delta$-temporal matching
which comes from the first $\Delta$-window, then 
we can exchange it with the first appearance of the edge.
\begin{lemma}%
	\label{lem:first-edge}
	Let $(G,\lambda)$ be a temporal graph and let $M$ be a $\Delta$-temporal matching in $(G,\lambda)$.
	Let also $e \in E_{t_1} \cap E_{t_2}$, where $t_1<t_2 \leq \Delta$.
	If $(e,t_1) \not \in M$ and $(e,t_2) \in M$,
	then $M' = (M \setminus \{ (e,t_2) \}) \cup \{ (e,t_1) \}$ is a $\Delta$-temporal matching in $(G,\lambda)$.
\end{lemma}
\begin{proof}
	The lemma follows from the observation that
	since $t_2 \leq \Delta$,
	no time edge $(e,t)$, $t < t_2$,
	is in conflict with any time edge in $M \setminus \{(e,t_2)\}$.
\end{proof}

We use \cref{lem:first-edge} to construct a small set $K$ of time edges 
from the first $\Delta$-window 
such that there exists a maximum $\Delta$-temporal matching~$M$ in $(G, \lambda)$
with the property that the restriction of $M$ to the first $\Delta$-window is contained in $K$.
\begin{definition}[Kernelization for the First $\Delta$-Window]
	Let $\Delta$ be a natural number and let $\G$ be a temporal graph.
	We call a set $K$ of time edges of $\G|_{[1,\Delta]}$ a \emph{kernel for the first $\Delta$-window of $\G$} 
	if %
	there exists a maximum $\Delta$-temporal matching $M$ in $\G$
	with~$M|_{[1,\Delta]} \subseteq K$.
\end{definition}
Informally, the idea for computing the kernel for the first $\Delta$-window is to first select vertices that are suitable to be matched. Then, for each of these vertices, we select the earliest appearance of a sufficiently large number of incident time edges, where each of these time edges corresponds to a different edge of the underlying graph. We show that we can do this in a way that the number of selected time edges can be bounded in the size $\nu$ of a maximum matching of the underlying graph $G$.
Formally, we aim at proving the following lemma.
\begin{lemma}%
	\label{lem:first-kernel}
	Given a natural number $\Delta$ and a temporal graph $\G = (G,\lambda)$,  
	we can compute in $O(\nu^2 \cdot |\G|)$ time
	a kernel $K$ for the first $\Delta$-window of~$\G$ 
	such that $|K| \in O(\nu^2)$.
\end{lemma}

\cref{alg:first-kernel} presents the pseudocode for the algorithm behind \cref{lem:first-kernel}.
\begin{algorithm2e}[t]

	Let $G'$ be the underlying graph of $\G|_{[1,\Delta]}$.\nllabel{line:ug} and $K = \emptyset$\;
	$A \gets $ a maximum matching of $G'$.\nllabel{line:max-matching}\;
	$V_A \gets $ the set of vertices matched by $A$.\;
  	\ForEach(\nllabel{line:for-Va}){$v \in V_A$}{
		$R_v \gets \big\{ (\{v,w\},t) \mid w \in N_{G'}(v)$ and $t = \min \{ i \in[\Delta] \mid \{v,w\} \in E_i \} \big\}$.\nllabel{line:Mv}\;
		\lIf{$|R_v| \leq 4\nu$}{
			$K \gets K \cup R_v$.
		}
		\uElse{
			Form a subset $R' \subseteq R_v$ such that $|R'| = 4\nu+1$ and 
			for every $(e,t) \in R'$ and $(e',t') \in R_v \setminus R'$ we have $t \leq t'$.\nllabel{line:large-set}\;
			$K \gets K \cup R'$.\;
		}
  	}
	\KwRet{$K$.}\;
	\caption{Kernel for the First $\Delta$-Window (\cref{lem:first-kernel}). }

  \label[algorithm]{alg:first-kernel}
\end{algorithm2e}
We show correctness of \cref{alg:first-kernel} in \cref{lem:first-kernel-correct} and examine its running time
in \cref{lem:first-kernel-runtime}.
Hence, \cref{lem:first-kernel} follows from \cref{lem:first-kernel-correct,lem:first-kernel-runtime}.
\begin{lemma}
	\label{lem:first-kernel-correct}
	\cref{alg:first-kernel} is correct, that is, it outputs a size-$O(\nu^2)$
	kernel $K$ for the first $\Delta$-window of $\G$.
\end{lemma}
\begin{proof}
	Let $M$ be a maximum $\Delta$-temporal matching of $\G$ such that
	$\big| M|_{[1,\Delta]} \setminus K\big|$ is minimized.
	Without loss of generality we can assume that every time edge in $M|_{[1,\Delta]}$ is
	the first appearance of an edge. Indeed, by construction, $K$ contains only the first appearances of edges,
	and therefore if $(e,t) \in M|_{[1,\Delta]}$ is not the first appearance of
	$e$, then by \cref{lem:first-edge} it can be replaced by the first appearance,
	and this would not increase $\big| M|_{[1,\Delta]} \setminus K\big|$.
	Now, assume towards a contradiction that $M|_{[1,\Delta]} \setminus K$ is not empty and 
	let $(e,t)$ be a time edge in $M|_{[1,\Delta]} \setminus K$.
	Since $A$ is a maximum matching in the underlying graph $G'$ of $\G|_{[1,\Delta]}$,
	at least one of the end vertices of $e$ is matched by $A$, i.e.,\ it belongs to
	$V_A$.
	Then for a vertex $v \in V_A \cap e$ we have that ${(e,t)} \in R_v$.
	Moreover, observe that $|R_v| > 4\nu$, because otherwise $(e,t)$ would be in $K$.
	For the same reason $(e,t) \not\in R'$, where
	$R' \subseteq R_v$ is the set of time edges computed in Line~\ref{line:large-set} of the algorithm.
	Let $W = \{ (w,t) \mid (\{v,w\},t) \in R' \}$ be the set of vertex appearances 
	which are adjacent to vertex appearance $(v,t)$ by a time edge in $R'$.
	Since $R_v$ contains only the first appearances of edges, we know that $W$ contains
	exactly $4\nu +1$ vertex appearances of pairwise different vertices.

	We now claim that $W$ contains a vertex appearance which is not $\Delta$-blocked by any time edge in $M$.
	To see this, we recall that $\nu$ is the maximum matching size of the underlying graph of $\G$.	
	Hence it is also an upper bound on the number of time edges in 
	$M|_{[1,\Delta]}$ and $M|_{[\Delta+1,2\Delta]}$, which implies that in the first $\Delta$-window
	vertex appearances of at most $4\nu$ distinct vertices are $\Delta$-blocked by time edges in $M$.
	Since $W$ contains $4\nu +1$ vertex appearances of pairwise different vertices, we conclude that there 
	exists a vertex appearance $(w',t') \in W$  which is not $\Delta$-blocked by~$M$.

	Observe that $t' \leq t$ because $(\{v,w'\},t') \in R'$ and $(e,t) \in R_v \setminus R'$.
	Hence, $(v,t')$ is not $\Delta$-blocked by $M \setminus \{ (e,t) \}$.
	Thus, $M^* := (M \setminus \{ (e,t) \}) \cup \{ (\{v,w'\},t') \}$
	is a $\Delta$-temporal matching of size $|M|$ with 
	$\big|M^*|_{[1,\Delta]} \setminus K \big| < \big| M|_{[1,\Delta]} \setminus K \big|$.
	This contradiction implies
	that $M|_{[1,\Delta]} \setminus K$ is empty and thus $M|_{[1,\Delta]} \subseteq K$.

	It remains to show that $|K| \in O(\nu^2)$.
	Since each maximum matching in $G'$ has at most $\nu$ edges, we have that $|V_A| \leq 2\nu$.
	For each vertex in $V_A$ the algorithm adds at most $4\nu+1$ time edges to $K$.
	Thus, $|K| \leq 2\nu\cdot(4\nu+1) \in O(\nu^2)$.
\end{proof}

\begin{lemma}%
	\label{lem:first-kernel-runtime}
	\cref{alg:first-kernel} runs in $O(\nu^2 (n + m \Delta))$ time.
	In particular, the time complexity of \cref{alg:first-kernel} is dominated by $O(\nu^2 |\G|)$.
\end{lemma}
\begin{proof}
	The underlying graph $G'$ of the first $\Delta$-window in Line \ref{line:ug} of \cref{alg:first-kernel}  
	can be computed in~$O(\Delta m)$ time.
	Using the standard augmenting path-based procedure and the linear-time algorithm
	for finding an augmenting path \cite{GT85}, a maximum matching $A$ of $G'$ in Line~\ref{line:max-matching}
	can be computed in $O(\nu(n+m))$ time.
	Since $|V_A| \leq 2\nu$, the for-loop in Line \ref{line:for-Va} performs at most~$2\nu$ iterations.
	At each of these iterations the corresponding set $R_v$ can be computed in $O(n)$ time,
	because it contains at most $n-1$ time edges, and the list of time labels of every edge is ordered by time.
	Finally, observe that $R' \subseteq R_v$ can be computed in $O(\nu\cdot n)$ time and 
	that at each iteration we add at most $4\nu+1$ time edges to $K$.
	Thus, overall \cref{alg:first-kernel} runs in $O(\nu^2 (n + m \Delta))$ time.
\end{proof}

Having \cref{alg:first-kernel} at hand, we can formulate a recursive search tree algorithm 
which (1) picks a time edge $(e,t)$ in the kernel of the first $\Delta$-window,
(2) removes all time edges which are not $\Delta$-independent with $(e,t)$, and
(3) calls itself until $k$ time edges are picked or the temporal graph is empty.

\begin{algorithm2e}[t]
	\KwIn{A temporal graph $\G = (G,\lambda)$ of lifetime $T$ and $\Delta,k \in \mathbb N$.}
	\KwOut{\yes{} if there is a $\Delta$-temporal matching of size $k$, otherwise \no.}
	
	\medskip

	\lIf{$k=0$ \textbf{\textup{or}} maximum matching size of $G$ is at least $k$ \nllabel{fpt-k:base-case1}}{ \KwRet{\yes}. }
	\lIf{$\G$ has no edge appearances}{ \KwRet{\no}.\nllabel{fpt-k:base-case2}} 
	Let $t_0$ be the time slot of the first non-empty snapshot of $\G$.\nllabel{fpt-k:t0}\;
	$\lambda(e) \gets \{ t-t_0+1 \mid t \in \lambda(e) \}$, for all $e \in E(G)$.\nllabel{fpt-k:shift}\;
	$K \gets$ kernel for the first $\Delta$-window of $\G$ computed by \cref{alg:first-kernel}.\nllabel{fpt-k:kernel}\;
	\ForEach(\nllabel{fpt-k:for}){time edge $(e,t) \in K$}{
		$A \gets \{ (e',t') \mid (e',t') \in \mathcal E(\G)$ is not $\Delta$-independent with $(e,t) \}$.\;
		\KwRet{\textbf{\textup{call}} \cref{alg:fpt-for-k} for $\G \setminus A$, $\Delta$, and $k \gets k-1$.\nllabel{fpt-k:call}}
	}

	\caption{Fixed-Parameter Algorithm for the Solution Size $k$ (\cref{thm:fpt-for-k}).}
  \label[algorithm]{alg:fpt-for-k}
\end{algorithm2e}
	The pseudocode for the algorithm behind \cref{thm:fpt-for-k} is stated in \cref{alg:fpt-for-k}.
We show its correctness in \cref{lem:fpt-for-k-correct} 
and the claimed running time in \cref{lem:fpt-for-k-runtime}.
\begin{lemma}
	\label{lem:fpt-for-k-correct}
	\cref{alg:fpt-for-k} is correct.
\end{lemma}
\begin{proof}
	First, observe that an instance with $k=0$ is a trivial \emph{yes}-instance
	and an instance with $k>0$ and no edge appearances is a trivial \emph{no}-instance.
	Second, if there is a matching~$M$ of size at least~$k$ in the underlying graph~$G$,
	then $\{ (e,t) \mid e \in M, t = \min \lambda(e) \}$ is a~$\Delta$-temporal matching in~$\G$
	of size $|M|$.
	Hence, Lines~\ref{fpt-k:base-case1}--\ref{fpt-k:base-case2} are correct.
	In Lines~\ref{fpt-k:t0}--\ref{fpt-k:shift}, we remove the leading edgeless snapshots from the temporal graph if any.
	Note that this does not change the size of any~$\Delta$-temporal matching.
	However, after this preprocessing
	every~$\Delta$-temporal matching~$M$ of maximum size in $\G$
	contains at least one time edge from the first~$\Delta$-window,
	because otherwise $M$ could be extended by a time edge from the first snapshot.
	In Line~\ref{fpt-k:kernel}, a kernel~$K$ for the first~$\Delta$-window of $\G$ is computed by \cref{alg:first-kernel}.
	Hence, there is a~maximum $\Delta$-temporal matchings $M$ in $\G$ 
	such that there is a time edge $(e,t)\in M|_{[1,\Delta]} \cap K$.
	Thus, at the iterations of the for-loop in Line~\ref{fpt-k:for}
	that correspond to $(e,t) \in M|_{[1,\Delta]} \cap K$, we have $A \cap M = \emptyset$.
	Consequently, since $\G \setminus A$ contains all time edges of $M \setminus \{ (e,t)\}$,
	there is a $\Delta$-temporal matching of size at least $k$ in $\G$ 
	if and only if 
	there is a $\Delta$-temporal matching of size at least~$k-1$ in~$\G \setminus A$.
	This implies correctness of Line~\ref{fpt-k:call}.

	\cref{alg:fpt-for-k} terminates because we decrease the parameter $k$ in each recursion until~zero is reached.
\end{proof}
It remains to show that \cref{alg:fpt-for-k} is indeed a linear-time
fixed-parameter algorithm when parameterized by the solution size $k$.
\begin{lemma}
	\label{lem:fpt-for-k-runtime}
	\cref{alg:fpt-for-k} runs in $k^{O(k)}\cdot |\G|$ time.
\end{lemma}
\begin{proof}
	In Line~\ref{fpt-k:base-case1} of \cref{alg:fpt-for-k}, we use the standard
	augmenting path-based algorithm for maximum matching to check
	if $G$ has a matching of size $k$.
	Since an augmenting path can be found in linear time \cite{GT85}, 
	this step can be executed in $O(k(n+m))$ time.
	If $G$ has a matching of size $k$, then the algorithm terminates in Line~\ref{fpt-k:base-case1} and the lemma holds.
	Hence, we assume that the maximum matching size $\nu$ of $G$ is strictly smaller than $k$.
	To compute Line~\ref{fpt-k:shift}, we first determine in linear time the time slot~$t_0$ of the first non-empty snapshot
	and then iterate a second time over the temporal graph to set the new labels. 
	By \cref{lem:first-kernel}, Line~\ref{fpt-k:kernel} can be computed in $O(\nu^2 \cdot |\G|)$ time.
	Thus, Lines~\ref{fpt-k:base-case1}--\ref{fpt-k:kernel} 
	are computable in $O(k^2 \cdot |\G|)$ time.

	By \cref{lem:first-kernel}, 
	the kernel $K$ for the first $\Delta$-window contains at most $O(k^2)$ time edges.
	Hence, the for-loop in Line \ref{fpt-k:for} runs at most $O(k^2)$ iterations.
	To compute the temporal graph~$\G \setminus A$ of Line~\ref{fpt-k:call} in $O(|\G|)$ time,
	we first iterate once over the temporal graph and remove all time edges which are not $\Delta$-independent with $(e,t)$. %
	In total, Lines \ref{fpt-k:base-case1}-\ref{fpt-k:call} of a single call (without the recursion) of \cref{alg:fpt-for-k} runs in $O(k^2\cdot |\G|)$ time.
	In Line~\ref{fpt-k:call} the algorithm calls itself recursively.
	However, since the parameter $k$ is decreased at every recursive call, the depth of the recursion tree is at most $k$,
	which implies that the size of the tree is $k^{O(k)}$. Hence \cref{alg:fpt-for-k} runs in $k^{O(k)}\cdot |\G|$ time.
\end{proof}

\subsection{Fixed-parameter tractability for the combined parameter $\Delta$ and maximum matching size $\nu$ of the underlying graph}\label{subs:fptdeltanu}
In this subsection, we show that \textsc{Temporal Matching} is fixed-parameter tractable
when parameterized by $\Delta$ and the maximum matching size $\nu$ of the underlying graph.
\begin{theorem}%
	\label{thm:fpt-for-alpha-delta}
	\textsc{Temporal Matching} can be solved 
	in $2^{O(\nu\Delta)}\cdot |\G| \cdot \frac{T}{\Delta}$ time.
\end{theorem}
The proof of \cref{thm:fpt-for-alpha-delta} is deferred to the end of this section.
Note that \cref{thm:fpt-for-alpha-delta} implies that \textsc{Temporal Matching} is fixed-parameter tractable
when parameterized by~$\Delta$ and the maximum matching size $\nu$ of the underlying graph, 
because there is a simple preprocessing step such that we can assume afterwards that the lifetime $T$ is polynomially bounded in the input size.
This preprocessing step modifies the temporal graph such that it does not contain $\Delta$ consecutive edgeless snapshots.
This can be done by iterating once over the temporal graph.
Observe, that this procedure does not change the maximum size of a $\Delta$-temporal matching and afterwards each $\Delta$-window contains at least one time edge.
Hence,~$\frac{T}{\Delta} \leq |\G|$.

Note that this result is incomparable to the result from the previous subsection (\cref{thm:fpt-for-k}). In some sense, we trade off replacing the solution size parameter $k$ with the structurally smaller parameter $\nu$ but we do not know how to do this without combining it with $\Delta$.
In comparison to the exact algorithm by Baste~et~al.~\cite{baste2018temporal} 
(who showed fixed-parameter tractability with $k$ and $\Delta$) 
we replace $k$ by the structurally smaller $\nu$, hence improving their result from a parameterized classification standpoint.
Furthermore, we note that \cref{thm:fpt-for-alpha-delta} is asymptotically optimal for any fixed $\Delta$ since 
there is no $2^{o(\nu)} \cdot |\G|^{f(\Delta,T)}$ algorithm for \textsc{Temporal Matching}, unless ETH fails (see \cref{th:apx}).

In the reminder of this section, we sketch the main ideas of the algorithm behind \cref{thm:fpt-for-alpha-delta}.
The algorithm works in three major steps: 
\begin{enumerate}
\item\label{step1} The temporal graph is divided into disjoint $\Delta$-windows,
\item\label{step2} for each of these $\Delta$-windows a small family of $\Delta$-temporal matchings is computed, 
and then 
\item\label{step3} the maximum size of a $\Delta$-temporal matching for the whole temporal graph is computed with a dynamic program.
\end{enumerate}

We first discuss how the algorithm performs Step~\ref{step2}.
Afterwards we formulate the dynamic program (Step~\ref{step3}) and prove \cref{thm:fpt-for-alpha-delta}.
In a nutshell, Step \ref{step2} consists of an iterative computation of
a small (upper-bounded in $\Delta+\nu$) family of $\Delta$-temporal matchings 
for an arbitrary $\Delta$-window such that at least one of 
them is ``extendable'' to a maximum $\Delta$-temporal matching for the whole temporal graph.
To this end, we introduce from matroid theory.
\ifarxiv{}
\subparagraph*{Tools from matroid theory.}
\fi{}
We use standard notation and terminology from matroid theory~\cite{Oxl92}.
  A pair $(U,I)$,
  where $U$~is the \emph{ground set}
  and $I\subseteq 2^U$ is a family of \emph{independent sets},
  is a \emph{matroid} if the following holds:
  \begin{itemize}
  \item  $\emptyset \in I$;
  \item  if $A' \subseteq A$ and $A \in I$, then $A' \in I$;
  \item  if $A,B \in I$ and $|A| < |B|$, then there is an $x \in B \setminus A$ such that $A \cup \{x\} \in I$.
  \end{itemize}
  An inclusion-wise maximal independent set~$A\in I$
  of a matroid~$\Q=(U,I)$ is a \emph{basis}.
  The cardinality of the bases of~$\Q$
  is called the \emph{rank} of~$\Q$.
  The \emph{uniform matroid of rank~$r$} on $U$
  is the matroid~$(U,I)$
  with $I=\{S\subseteq U\mid |S|\leq r\}$.
A matroid~$(U,I)$ is \emph{linear} or \emph{representable over a field $\mathbb F$}
if there is a matrix~$A$ with entries in $\mathbb F$ and the columns labeled by the elements of~$U$
such that $S \in I$ if and only if the columns of~$A$ 
with labels in~$S$ are linearly independent over~$\mathbb F$. The matrix $A$ is called a \textit{representation}
of $(U,I)$.

\begin{definition}[Max $q$-Representative Family]\label[definition]{def:qrep}
  Given a matroid $(U,I)$,
  a family $\mathcal S \subseteq I$ of independent sets,
  and a function~$w\colon \mathcal S \rightarrow \mathbb R$,
  we say that a subfamily~$\widehat{\mathcal S} \subseteq \mathcal S$
  is a \emph{max $q$-representative
  for $\mathcal S$ with respect to~\(w\)} if
  for each set~$Y \subseteq U$ of size at most~$q$
  it holds that
  if there is a set~$X \in \mathcal S$  with
  $X \uplus Y \in I$,  then there is a
  set~$\widehat X \in \widehat{\mathcal S}$ such that  $\widehat X \uplus Y \in I$ and
  $w(\widehat X) \geq w(X)$.
\end{definition}

For linear matroids, there are fixed-parameter algorithms parametrized by rank that compute
representatives for large families of independent sets with respect to additive set functions~\cite{ROP-Arxiv18}.
A function $w : 2^U \rightarrow \mathbb{R}$ on the subsets of a set $U$ 
is \emph{additive set function} if $w(A \uplus B) = w(A) + w(B)$ for all disjoint sets $A,B \subseteq U$. 

\begin{theorem}[van Bevern et al.\ {\cite[Proposition 4.8]{ROP-Arxiv18}}]
	\label{thm:matroid-tool}
  Let $\alpha$, $\beta$, and $\gamma$ be non-negative integers such that $r=(\alpha+\beta)\gamma \geq 1$.
  Let $\Q = (U,I)$~be a linear matroid of rank $r$ and $w \colon 2^U \to \mathbb N$~be an additive set function.
  Furthermore, let  $\mathcal H\subseteq 2^U$~be a $\gamma$-family of size~$t$ and let
  \[\mathcal S = \{ S = H_1 \uplus \dots \uplus H_\alpha 
    \mid
    S\in I\text{ and }
    H_j \in \mathcal H\text{ for }j \in \{1,\dots, \alpha \}\}.
  \]
  Then, given a representation of~$\Q$ over a finite field~$\mathbb F$, one can
  compute a max $\beta \gamma$-representative~$\widehat{\mathcal S}$ of size~$r \choose \alpha \gamma$ for the 
  family $\mathcal{S}$ with respect to~\(w\)
  using $2^{O(r)}\cdot t$~operations over~$\mathbb F$
  and calls to the function $w$.
  \end{theorem}
  \cref{thm:matroid-tool} is based on results of Fomin et al.\ \cite{FLPS16} and Marx \cite{Mar09}.
  We use \cref{thm:matroid-tool} only for uniform matroids.
  For this reason we expect that
  one can improve the base of the exponential function 
  in $\nu\Delta$ of the running time in \cref{thm:fpt-for-alpha-delta} 
  by replacing Theorem 1.1 of Fomin et al.\ \cite{FLPS16} for linear matroids 
  with its special case Theorem 1.2 for uniform matroids and 
  tighten the running time analysis in \cref{thm:matroid-tool}.
  However, the dependence on the size of the temporal graph in the running time would be worth.

  Furthermore, van Bevern et al.\ \cite{ROP-Arxiv18} proved \cref{thm:matroid-tool}
  for multiple matroids and for more general weight functions than
  additive set functions.
  However, for our purpose the stated version suffices.
  The crucial point of \cref{thm:matroid-tool} is that 
  for a linear matroid of rank $(\alpha+\beta)\gamma$ and a $\gamma$-family $\mathcal H$, 
  we can compute a small (of size $r \choose \alpha\gamma$)~max $\beta\gamma$-representative $\widehat{\mathcal{S}}$
  for a potentially very large (unbounded in the rank of the matroid) family~$\mathcal S$ of all independent sets of size $\alpha \gamma$
  which are disjoint unions of sets from $\mathcal H$.
  An important property of $\widehat{\mathcal S}$ is that for any independent set $Y$ of size $\beta\gamma$
  such that there is a set $X \in \mathcal S$ 
  which is disjoint from $Y$ and the union of $X$ and $Y$ is an independent set, 
  $\widehat{\mathcal S}$ contains a set~$\widehat X$
  which is also disjoint from $Y$, the union of $\widehat X$ and $Y$ is also an independent set, 
  and the weight of $\widehat X$ is at least as large as the weight of $X$.

  \ifarxiv{}
\subparagraph*{Families of $\ell$-complete $\Delta$-temporal matchings.}
\fi{}
	Throughout this section
	let $\G = (G=(V,E),\lambda)$ be a temporal graph of lifetime $T$ and let $\nu$ be the maximum matching size in $G$.
	Let also $\Delta$ and $\ell$ be natural numbers such that $\ell\Delta \leq T$.

	A family $\mathcal M$ of $\Delta$-temporal matchings of $\G|_{[\Delta(\ell-1)+1, \Delta\ell]}$ is called~\emph{$\ell$-complete}
	if for any $\Delta$-temporal matching $M$ of $\G$ there is $M' \in \mathcal M$ 
	such that $\big(M \setminus M|_{[\Delta(\ell-1)+1, \Delta\ell]}\big) \cup M'$ is a $\Delta$-temporal matching of $\G$ of
	size at least $|M|$.
	A central part of our algorithm is an efficient procedure for computing an $\ell$-complete family.
	Formally, we aim for the following lemma.
\begin{lemma}%
	\label{lem:comp-ext-fam}
	There exists a $2^{O(\nu\Delta)} \cdot |\G|$-time algorithm to compute
	an $\ell$-complete family of size 
	$2^{O(\nu \Delta)}$
	of $\Delta$-temporal matchings of $\G|_{[\Delta(\ell-1)+1, \Delta\ell]}$.
\end{lemma}
\noindent
In the proof of \cref{lem:comp-ext-fam} we employ representative families and other tools from matroid theory.
Before we show \cref{lem:comp-ext-fam}, we prove several intermediate lemmas.
These lemmas are used in the proof of \cref{lem:comp-ext-fam} 
which is deferred to the end of this paragraph.
The primary tool in the proof of \cref{lem:comp-ext-fam} is \cref{thm:matroid-tool}
applied to a properly chosen matroid $\Q$, a family~$\mathcal H$, and a weight function $w$.
The idea is that a disjoint union of sets from $\mathcal H$ corresponds to 
a $\Delta$-temporal matching in $\G|_{[\Delta(\ell-1)+1, \Delta\ell]}$ and the weight function tells 
us how large the $\Delta$-temporal matching is.
\begin{construction}[Matroid, Family, and Weight Function]
	\label{const:matroid}
	We define
	\begin{enumerate}
		\item the $\left(5\nu+5\nu(\Delta-1)\right)$-uniform \textit{matroid} $\Q$ on the ground set
		$U := V \cup E' \cup V' \cup D$, where 
		\begin{itemize}
			\item $E' := \big\{ e_t \mid e \in E_t$ and $ t \in
			[\Delta(\ell-1)+1,\Delta\ell] \big\}$,
			\item $V' := \big\{ v_t \mid v \in V $ and $ t \in
			[\Delta(\ell-1)+1,\Delta\ell] $ and $v$ is not isolated in $G_t \big\}$, and
			\item $D := \big\{ d_i \mid i \in [5\nu] \big\}$;
		\end{itemize}
		
		\item a $5$-\textit{family} $\mathcal H := \mathcal H_E \cup \mathcal H_D$, where
		\begin{itemize}
			\item $\mathcal H_E := \big\{ E^{(t)}_{\{v,w\}} = \{ v,w, v_t,w_t,e_t\} \mid
			e= \{v,w\} \in E_t$ and $t \in [\Delta(\ell-1)+1,\Delta\ell] \big\}$, and
			\item $\mathcal H_D := \big\{ D_i = \{ d_{5(i-1)+j} \mid j \in [5]\} \mid i
			\in [\nu] \big\}$;
		\end{itemize}
		
		\item a \textit{weight function} $w : 2^U \to \mathbb N; X \mapsto |X \cap E'|$.
	\end{enumerate}
	\hfill$\lhd$
\end{construction}

Observe that each set in $\mathcal H_E$ corresponds to a time edge of the temporal graph.
Furthermore, $D$ is the set of \textit{dummy} elements and $\mathcal H_D$ is a family of sets of dummy elements,
which we introduce for technical reasons in order to able to apply \cref{thm:matroid-tool} 
and they can be ignored for the moment.

An important property of \cref{const:matroid} that we will employ in the proof of \cref{lem:comp-ext-fam} 
is formalized in the following simple observation.

\begin{observation}
	\label{lem:const-correct}
	Let $M$ be a set of time edges in $\G|_{[\Delta(\ell-1)+1,\Delta\ell]}$.
	Then $M$ is a $\Delta$-temporal matching in $\G|_{[\Delta(\ell-1)+1,\Delta\ell]}$
	if and only if the sets $E^{(t)}_{e}$, $(e,t) \in M$, are pairwise disjoint.
\end{observation}

Before we proceed to the proof of \cref{lem:comp-ext-fam}, we show that both a representation of the matroid $\Q$
and the family $\mathcal H$ can be computed efficiently.

\begin{lemma}%
	\label{lem:eff-const}
	A representation of the matroid $\Q$ 
	over a finite field $\mathbb F_p$ with $p \in O(|\G|)$
	and the family $\mathcal H$ can be computed
	in $O(\nu \Delta |\G|)$ time.
	Furthermore, one operation over the finite field $\mathbb F_p$ can be computed in constant time.
\end{lemma}
\begin{proof}
	Observe that the ground set $U$ is of size $O(|\G|)$, %
	because $E'$ contains one element for each time edge in $\G|_{[\Delta(\ell-1)+1,\Delta\ell]}$,
	the cardinality of $V'$ is upper-bounded by $2|E'|$, 
	and the set $D$ contain $O(\nu)$ dummy elements, which is dominated by $O(|\G|)$.
	The set $U$ can be constructed in $O(|\G|)$ time by iterating once over the temporal subgraph $\G|_{[\Delta(\ell-1)+1,\Delta\ell]}$
	and adding the dummy elements afterwards.

	Let $s = |U|$ and $r = 5\nu + 5\nu(\Delta-1)$.
	Let $p$ be a prime number with $s \leq p \leq 2s$. Such a prime number exists by the folklore Bertrand-Chebyshev theorem
	and can be found in $O(s^{1/2+o(1)}) \subseteq O(s)$ time using the Lagarias-Odlyzko method \cite{tao2012deterministic}.
	For the sake of completeness, we now show that there exists a representation of $\Q$ over $\mathbb F_p$ and it can be computed in 
	$O(\nu \Delta |\G|)$ time.
	Note that this is a detailed version of the proof of Marx \cite[Section~3.5]{Mar09}.
	Let $x_1, x_2, \ldots, x_s$ be distinct elements of $\mathbb F_p$. 
	We define an $r \times s$ matrix 
	$$
	A =
		\begin{pmatrix}
    			1 & 1 & 1 & \dots  & 1 \\
    			x_{1} & x_{2} & x_{3} & \dots  & x_{s} \\
    			x_{1}^2 & x_{2}^2 & x_{3}^2 & \dots  & x_{s}^2 \\
    			\vdots & \vdots & \vdots & \ddots & \vdots \\
    			x_{1}^{r-1} & x_{2}^{r-1} & x_{3}^{r-1} & \dots  & x_{s}^{r-1}
		\end{pmatrix}
	$$
	and claim that $A$ is a desired representation of $\Q$. Indeed, 
	any $r$ columns (corresponding to $x_{i_1}, x_{i_2}, \ldots, x_{i_r}$)
	of $A$ form a Vandermonde matrix, 
	whose determinant is $\prod_{1 \leq j < k \leq r} (x_{i_j} - x_{i_k}) \neq 0$. 
	Therefore, any $r$-element subset of the ground set is linearly independent, 
	while clearly any larger subsets are linearly dependent.

	To perform a primitive operation in the finite field $\mathbb F_p$, 
	we first perform the primitive operation in $\mathbb Z$ 
	and then take the result modulo $p$.
	Hence, we can compute one operation over $\mathbb F_p$ in constant time,
	because $p \in O(s)$ is small and we assume the RAM model of computation (see
	\cref{sec:preliminaries}).
	
	Notice now that every element of matrix~$A$ is either $1$ or can be obtained 
	by a single multiplication from the previous element in its column. 
	Hence, $A$ can be computed in $rs \in O(\nu \Delta |\G|)$ time.
		
	Finally, the family $\mathcal H$ can be computed in $O(|\G|)$ time by iterating once over $E'$ (to create~$\mathcal H_E$)
	and then adding the dummy sets of $\mathcal H_D$.
	Thus, overall a representation of $\mathcal Q$ and the family $\mathcal H$ can be computed in
	$O(\nu \Delta |\G|)$ time.
\end{proof}

Now we are ready to prove \cref{lem:comp-ext-fam}.
The algorithm behind \cref{lem:comp-ext-fam} is stated in \cref{alg:comp-ext-fam}.
\begin{algorithm2e}[t]
	\KwIn{A temporal graph $\G = (G,\lambda)$ of lifetime $T$ and $\ell, \Delta \in \mathbb{N}$ such that $\ell\Delta \leq T$.}
	\KwOut{An $\ell$-complete family of $\Delta$-temporal matchings for $\G|_{[\Delta(\ell-1)+1,\Delta\ell]}$ of size 
	$2^{O(\nu \Delta)}$.}

	\medskip
	
	$\nu \gets$ the maximum matching size of $G$.\nllabel{line:comp-max-matching}\;
	For input $\G$, $\Delta$, and $\nu$, compute a representation of the matroid $\mathcal Q = (U,I)$ 
	over a finite field $\mathbb F_p$ with $p \in O(|\G|)$, and the family $\mathcal H$ according to \cref{const:matroid}.\nllabel{line:const}\; 
	$\widehat{\mathcal F} \gets$
	\hangindent=2.5\skiptext\hangafter=1
	max $(5\nu(\Delta-1))$-representative family of 
			$ \mathcal F = \big\{ F = H_1 \uplus \dots \uplus H_\nu \mid
				F\in I\text{ and }
			H_j \in \mathcal H\text{ for }j \in [\nu] \big\}
				$ with respect to $w$.\nllabel{line:comp-rep-fam}\;
			$\mathcal M \gets \emptyset$.\;
		\ForEach(\nllabel{comp-ext-fam:for}){$F \in \widehat{\mathcal F}$}{
		$M \gets \big\{ (e,t) \mid e_t \in F \big\}$.\;
		$\mathcal M \gets \mathcal M \cup \{ M \}$.\;
	}
	\KwRet{$\mathcal M$.}
	\caption{Construction of an $\ell$-Complete Family (\cref{lem:comp-ext-fam}).}
  \label[algorithm]{alg:comp-ext-fam}
\end{algorithm2e}
Observe that in the following proof we will use the dummy elements, introduced in \cref{const:matroid}, to fill up the sets 
such that their sizes match the rank of the matroid $\mathcal Q$.
\begin{proof}[Proof of \cref{lem:comp-ext-fam}]
	To prove the lemma we use \cref{alg:comp-ext-fam}.
	We start with the running time analysis of \cref{alg:comp-ext-fam}.
	\begin{enumerate}%
		\item To compute the maximum matching size $\nu$ of the underlying graph $G$,
		we use the standard augmenting path-based algorithm for maximum matching.
		Since an augmenting path can be found in linear time \cite{GT85}, 
		the computation of $\nu$ in Line \ref{line:comp-max-matching} takes $O(\nu|\G|)$ time.
		\item In Line \ref{line:const} the algorithm computes a representation of the matroid $\mathcal Q$ 
		over a finite filed $\mathbb F_p$ with
		$p \in O(|\G|)$ and the family $\mathcal H$ from \cref{const:matroid}.
		By \cref{lem:eff-const}, this can be done in $O(\nu \Delta |\G|)$ time.

		\item Since the rank of $\mathcal Q$ is $5(\nu+\nu(\Delta-1))$ and
		$|\mathcal H| \in O(|\G|)$, by \cref{thm:matroid-tool}, 
		the computation of a max $(5\nu(\Delta-1))$-representative family $\widehat{\mathcal F}$
		in Line~\ref{line:comp-rep-fam} performs 
		$2^{O(\nu\Delta)} \cdot |\G|$ operations in~$\mathbb F_p$ and calls to the function $w$.
		The  algorithm behind \cref{thm:matroid-tool} evaluates function~$w$ on sets of 
		cardinality at most the rank of~$\mathcal Q$, and hence a single call to the function $w$ from \cref{const:matroid}
		can be implemented to work in $O(\nu \Delta)$ time.
		Furthermore, by \cref{lem:eff-const} a single operation in $\mathbb F_p$ takes constant time.
		Hence, the overall time complexity of Line~\ref{line:comp-rep-fam} is $2^{O(\nu\Delta)} \cdot |\G|$.
		\item Since the family $\widehat{\mathcal F}$ is of size at most 
		${{5\nu+5\nu(\Delta-1)} \choose {5\nu}} \in 2^{O(\nu \Delta)}$,
		the for-loop in Line~\ref{comp-ext-fam:for} runs 
		$2^{O(\nu \Delta)}$
		iterations.
		Each of these iterations runs in $O(\nu)$ time, and hence, in total, the for-loop is executed in 
		$2^{O(\nu \Delta)}$ time.
	\end{enumerate}
	Overall the algorithm outputs a family $\M$ of size 
	$2^{O(\nu \Delta)}$
	in time~$2^{O(\nu \Delta)}\cdot |\G|$.

	We are left to show that $\mathcal M$ is an $\ell$-complete family of $\Delta$-temporal
	matchings of $\G|_{[\Delta(\ell-1)+1,\Delta\ell]}$.
	First, we argue that every set in $\mathcal M$ is a $\Delta$-temporal matching of $\G|_{[\Delta(\ell-1)+1,\Delta\ell]}$.
	Indeed, by construction, a set $M$ in $\mathcal M$ corresponds to a set $F$ in $\widehat{\mathcal F}$ that
	contains $\biguplus_{(e,t) \in M} E_e^{(t)}$ as a subset. Hence, by \cref{lem:const-correct}, the set $M$ is a 
	$\Delta$-temporal matching of $\G|_{[\Delta(\ell-1)+1,\Delta\ell]}$.

	We now show that $\mathcal M$ is $\ell$-complete.
	Let $M$ be a $\Delta$-temporal matching of $\G$, $M^{\ell} = M|_{[\Delta(\ell-1)+1,\Delta\ell]}$, and 
	$M' = M \setminus M^{\ell}$.
	Let also $W$ be the set of vertex appearances in $\G|_{[\Delta(\ell-1)+1,\Delta\ell]}$ 
	which are $\Delta$-blocked by $M'$.
	Note that since $M$ is a $\Delta$-temporal matching, no time edge in $M^{\ell}$ is incident with a
	vertex appearance in $W$. The latter together with \cref{lem:const-correct} imply that 
	the sets $Y = \{ v_t \in U \mid (v,t) \in W \}$ and $E_e^{(t)}, (e,t) \in
	M^{\ell}$, are pairwise disjoint.
	Since the maximum matching size of the underlying graph $G$ is $\nu$,
	we have that $|Y| = |W| \leq 4\nu(\Delta-1)$.
	For the same reason $|M^{\ell}| \leq \nu$ and therefore $\mathcal F$ contains a set 
	$X = \biguplus_{(e,t) \in M^{\ell}} E^{(t)}_{e} \uplus D'$ of size $5\nu$, where $D'$ is a set of dummy elements. 
	Consequently, the cardinality of $X \uplus Y$ is at most $5\nu+4\nu(\Delta-1)$
	and hence $X \uplus Y$ is an independent set of $\Q$.
	Furthermore, observe that $w(X) = |M|$.
	Now, since $\widehat{\mathcal F}$ is a max $(5\nu(\Delta-1))$-representative of $\mathcal F$ with respect to $w$,
	the family $\widehat{\mathcal F}$ contains a set $\widehat X$ such that $\widehat X$ is disjoint from $Y$,
	the union $\widehat X \uplus Y$ is an independent set of $\mathcal Q$, and $w(\widehat X) \geq w(X)$.
	Let $\widehat X'$ be the set obtained from $\widehat X$ by removing the dummy elements.
	Hence $w(\widehat X') = w(\widehat X)$ and by construction $\widehat X'$ is the union of pairwise disjoint sets $E^{(t)}_{e}$, $(e,t) \in M''$ for some set $M''$
	of time edges of $\G|_{[\Delta(\ell-1)+1,\Delta\ell]}$. 
	Thus, $w(\widehat X') = |M''|$.
	By \cref{lem:const-correct} we conclude that $M''$ is a
	$\Delta$-temporal matching of $\G|_{[\Delta(\ell-1)+1,\Delta\ell]}$. Moreover, no time edge in $M''$ is 
	incident with vertex appearances in $W$, as $\widehat X'$ is disjoint from $Y$.
	Hence, $M' \cup M''$ is a $\Delta$-temporal matching in $\G$ and 
	$|M' \cup M''| =  |M'| + |M''| = |M'| + w(\widehat X) \geq |M'| + w(X) = |M'| + |M^{\ell}| = |M|$.

\end{proof}

\ifarxiv{}
\subparagraph*{Dynamic program.}
\fi{}
Now we are ready to combine Step~\ref{step2} of our algorithm with the remaining Steps~\ref{step1} and~\ref{step3}.
More precisely, we employ $\ell$-complete families 
of $\Delta$-temporal matchings of $\Delta$-windows in a dynamic program (Step~\ref{step3}) 
to compute the $\Delta$-temporal matching of maximum size for the whole temporal graph.
The pseudocode of this dynamic program is stated in \cref{alg:dynamic-program}.
This is the algorithm behind \cref{thm:fpt-for-alpha-delta}.
It computes a table~$\T$ where each entry $\T[i,M']$ stores
the maximum size of a $\Delta$-temporal matching~$M$ in the temporal graph $\G|_{[1,\Delta i]}$
such that $M|_{[\Delta(i-1)+1,\Delta i]} = M'$.
Observe that a trivial dynamic program which computes all entries of~$\T$ cannot provide fixed-parameter tractability
of \textsc{Temporal Matching} when parameterized by $\Delta$ and $\nu$ because the corresponding table is simply too large.
The crucial point of the dynamic program is 
that it is sufficient to fix for each $i \in \frac{T}{\Delta}$ an $i$-complete 
family $\mathcal M_i$ of $\Delta$-temporal matchings for $\G|_{[\Delta(i-1)+1,\Delta i]}$
and then compute only the entries $\mathcal T[i,M']$, where~$M' \in \mathcal M_i$.

\begin{algorithm2e}[t]
	\KwIn{A temporal graph $\G = (G,\lambda)$ of lifetime $T$ and an integer $\Delta < T$.}
	\KwOut{The maximum size of a $\Delta$-temporal matching in $\G$.}
		
	\medskip
	$\T[i,M'] \gets 0 $, for every $i \in [\frac{T}{\Delta}]$ and a subset $M'$ of time edges of $\G|_{[\Delta(i-1)+1,\Delta i]}$.\nllabel{dp:init}\;

	$\mathcal M_0 \gets \{ \emptyset \}$.\;
	\For(\nllabel{dp:fori}){$i\gets1$ \KwTo $\frac{T}{\Delta}$}{ 
		$\mathcal M_i \gets$ $i$-complete family of $\Delta$-temporal matchings of $\G|_{[\Delta(i-1)+1,\Delta i]}$\nllabel{dp:comp-fam}.\;	
		\ForEach(\nllabel{dp:for}){$M_L \in \mathcal M_{i-1}$ and $M_R \in \mathcal M_i$}{
			\uIf{$M_L \cup M_R$ is a $\Delta$-temporal matching in $\G$}{
				$\T[i,M_R] \gets \max\big\{\T[i,M_R],\T[i-1,M_L]+|M_R|\big\}$.\nllabel{dp:add}\;
			}	
		}
	}
	\KwRet{$\displaystyle \max_{M' \in \mathcal M_{\frac{T}{\Delta}}} \T[\frac{T}{\Delta},M']$.\nllabel{dp:return}}
	\caption{Fixed-Parameter Algorithm for the Combined Parameter $\Delta$ and Maximum Matching Size $\nu$ of the Underlying Graph (\cref{thm:fpt-for-alpha-delta}).}
  \label[algorithm]{alg:dynamic-program}
\end{algorithm2e}

\begin{lemma}
	\label{lem:dynamic-program-correctness}
	\cref{alg:dynamic-program} is correct, that is, for a given temporal graph $\G = (G,\lambda)$ of lifetime $T$ 
	and an integer $\Delta < T$, the algorithm returns the maximum size of a $\Delta$-temporal matching in $\G$.
\end{lemma}
\begin{proof}
	To prove the lemma we first show by induction on $i \in [\frac{T}{\Delta}]$ that 
	for every $M' \in \M_i$ the entry $\T[i,M']$ contains the maximum size of a $\Delta$-temporal matching~$M$ in $\G|_{[1,\Delta i]}$ such that $M|_{[\Delta(i-1)+1,\Delta i]} = M'$.
	The statement is easily verifiable for $i=1$. 
	
	Let now $i \geq 2$ and assume the statement holds for indices smaller than $i$.
	Let $M^i$ be an arbitrary element in $\M_i$ and assume towards a contradiction that there is 
	a $\Delta$-temporal matching $M$ in $\G|_{[1,\Delta i]}$ 
	such that $|M|> \T[i,M^i]$ and $M|_{[\Delta (i-1) + 1,\Delta i]} = M^i$.
	
	Since $\M_{i-1}$ is an $(i-1)$-complete family of $\Delta$-temporal matchings 
	of the temporal graph $\G|_{[\Delta (i-2)+1,\Delta (i-1)]}$, there exists an $M^{i-1} \in \M_{i-1}$
	such that $M' := \big( M \setminus M|_{[\Delta (i-2)+1,\Delta (i-1)]} \big) \cup M^{i-1}$ 
	is a $\Delta$-temporal matching and $|M'| \geq |M|$.

	Since $M'|_{[\Delta (i-2)+1, \Delta (i-1)]} =  M^{i-1}$, by the induction hypothesis we have
	$\T[i-1, M^{i-1}] \geq \big| M'|_{[1,\Delta (i-1)]} \big|$.
	Furthermore, since both $M^{i-1}$ and $M^i$ are subsets of~$M'$, their union $M^{i-1} \cup M^i$
	is a $\Delta$-temporal matching in $\G$.
	Consequently, Line~\ref{dp:add} of the algorithm implies that 
	$\T[i,M^i] \geq \T[i-1, M^{i-1}] + |M^i|$, and therefore 
	$|M| > \T[i,M^i] \geq \big| M'|_{[1,\Delta (i-1)]} \big| + |M^i| = |M'|$,
	which is a contradiction.	
	
	To complete the proof, we observe that since $\M_{\frac{T}{\Delta}}$ is a $\frac{T}{\Delta}$-complete
	family of $\Delta$-temporal matchings of $\G|_{[T-\Delta+1,T]}$, the above statement implies that 
	the value $\max_{M' \in \mathcal \M_{\frac{T}{\Delta}}} \T[\frac{T}{\Delta},M']$ returned by the algorithm is the
	size of a maximum $\Delta$-temporal matching of $\G$.

\end{proof}

Next, we analyze the running time of the algorithm.
\begin{lemma}
	\label{alg:dynamic-program-runtime}
	\cref{alg:dynamic-program} runs in $2^{O(\nu\Delta)}\cdot |\G| \cdot \frac{T}{\Delta}$ time,
	where $\nu$ is the maximum matching size of underlying graph of $\G$.
\end{lemma}
\begin{proof}
	We represent our table $\T$ by a sparse set \cite{BT93} that stores only non-zero entires of $\T$.
	Hence, Line~\ref{dp:init} can be computed in constant time.
	By \cref{lem:comp-ext-fam}, Line~\ref{dp:comp-fam} 
	can be computed in 
	$2^{O(\nu \Delta)}\cdot |\G|$ time 
	and $|\mathcal M_i| \in 2^{O(\nu\Delta)}$.
	The latter implies that the for-loop of Line~\ref{dp:for} executes $2^{O(\nu\Delta)}$ iterations.
	Furthermore, each of the iterations runs in $O(\nu)$ time.
	Hence, all in all, \cref{alg:dynamic-program} runs in~$2^{O(\nu\Delta)}\cdot |\G| \cdot \frac{T}{\Delta}$ time.
\end{proof}
Finally, we have everything at hand to show \cref{thm:fpt-for-alpha-delta}.
\begin{proof}[Proof of \cref{thm:fpt-for-alpha-delta}]
	Let $(\G,\Delta,k)$ be an instance of \textsc{Temporal Matching}, where 
	$\G$ is a temporal graph of lifetime $T$ and $\Delta,k \in \mathbb N$.
	
	If $\Delta \geq T$, then we check whether the underlying graph $G$ of $\G$
	admits a matching of size at least $k$, which can be done in $O(k|\G|)$ time using 
	the standard augmenting path-based method.

	If $\Delta < T$, then we add at most $\Delta-1$ trailing edgeless snapshots to $\G$ to guarantee that the lifetime
	of the resulting temporal graph is divisible by $\Delta$. 
	Note that this does not change the maximum size of a $\Delta$-temporal matching. 
	We then apply \cref{alg:dynamic-program} to find the maximum size of a 
	$\Delta$-temporal matching in $\G$ and compare the resulting value with $k$.
	By \cref{alg:dynamic-program-runtime} this can be done in $2^{O(\nu\Delta)}\cdot |\G| \cdot \frac{T}{\Delta}$ time,
	which implies the theorem.
\end{proof}

\ifarxiv{}
\subparagraph*{Kernelization lower bound.}
\fi{}
Lastly, we can show that we cannot hope to obtain a polynomial kernel for the
parameter combination number $n$ of vertices and $\Delta$. In particular, this
implies that we presumably also cannot get a polynomial kernel for the parameter combination $\nu$ and $\Delta$, since $\nu\le \frac{n}{2}$.

\begin{proposition}%
		\label{prop:nopk}
\textsc{Temporal Matching} parameterized by the number $n$ of vertices 
does not admit a polynomial kernel for all $\Delta\ge 2$, unless NP $\subseteq$ coNP/poly.
\end{proposition}
 We need the following notation for the proof.
An equivalence
relation~$R$ on the instances of some problem~$L$ is a
\emph{polynomial equivalence relation} if
\begin{enumerate}[(i)]
 \item one can decide for each two instances in time polynomial in their sizes whether they belong to the same equivalence class, and
 \item for each finite set~$S$ of instances, the relation $R$ partitions the set into at most~$(\max_{x \in S} |x|)^{O(1)}$ equivalence classes.  
\end{enumerate}

An \emph{AND-cross-composition} of a problem~$L\subseteq \Sigma^*$ into a
parameterized problem~$P$ (with respect to a polynomial equivalence
relation~$R$ on the instances of~\(L\)) is an algorithm that takes
$\ell$ $R$-equivalent instances~$x_1,\ldots,x_\ell$ of~$L$ and
constructs in time polynomial in $\sum_{i=1}^\ell |x_i|$ an instance
$(x,k)$ of~\(P\) such that
\begin{enumerate}[(i)]
\item $k$ is polynomially upper-bounded in $\max_{1\leq i\leq \ell}|x_i|+\log(\ell)$ and 
\item $(x,k)$ is a yes-instance of $P$ if and only if $x_{\ell'}$ is a yes-instance of $L$ for every $\ell'\in \{1,\ldots,\ell\}$. 
\end{enumerate}

If an NP-hard problem~\(L\) AND-cross-composes into a parameterized
problem~$P$, then~$P$ does not admit a polynomial-size kernel, unless NP $\subseteq$ coNP/poly~\cite{bodlaender2014kernelization},
which would cause a collapse of the polynomial-time hierarchy to the third
level.

\begin{proof}[Proof of \cref{prop:nopk}]
		To prove \cref{prop:nopk}, we provide an AND-cross-composition from \textsc{Independent Set} on graphs with maximum degree three~\cite{garey1974some}.
Intuitively, we can just string together instances produced by \cref{constr:lreduction} in the time axis such that the large instance contains a large $\Delta$-temporal matching if and only if all original instances are \yes-instances.
 
In this problem we are asked to decide whether a given graph $H = (U,F)$ with maximum degree three contains a set of at least $h$ pairwise non-adjacent vertices.
Furthermore, it is important to observe that, given graph $H = (U,F)$ with maximum degree three, it is NP-complete to decide whether $H$ contains an independent set of size~$h$ even if it is known that $H$ does not contain an independent set of size $h+1$~\cite{garey1974some}. In the following, we assume that all instances have this property.
We define an equivalence relation~$R$ as follows: Two instances $(H=(U,F),h)$ and $(H'=(U',F'),h')$ are equivalent under~$R$ if and only if the number of vertices is the same, that is, $|U|=|U'|$ and we have that $h=h'$. Clearly,~$R$~is a polynomial equivalence relation. 

Now let $(H_1=(U_1,F_1),h_1),\ldots,(H_\ell=(U_\ell,F_\ell),h_\ell)$ be $R$-equivalent instances of \textsc{Independent Set} with the above described extra conditions. We arbitrarily identify the vertices of all instances, that is, let $U=U_1=\ldots=U_\ell$. For each $(H_i,h_i)$ with $i\in[\ell]$ we construct an instance of \textsc{Temporal Matching} as described in \cref{constr:lreduction} (for an illustration see \cref{fig:APXreduction}) with the only difference that we add a fourth snapshot that does not contain any edges. Now we put all constructed temporal graphs next to each other in temporal order, that is, if $\G^{(i)}=(G^{(i)}=(V^{(i)},E^{(i)}),\lambda^{(i)})$ with $\lambda^{(i)}:E^{(i)}\rightarrow [4]$ is the graph constructed for $(H_i,h_i)$, then the overall temporal graph is $\G=(G(\bigcup_{i\in[\ell]}V^{(i)},\bigcup_{i\in[\ell]}E^{(i)}),\lambda)$ with $\lambda(e) =\bigcup_{i\in[\ell]} \lambda^{(i)}(e)$, where we assume that $\lambda^{(i)}(e)=\emptyset$ if $e\notin E^{(i)}$. Note that $|\bigcup_{i\in[\ell]}V^{(i)}|\le 2|U|+\binom{|U|}{2}$ since the temporal graphs produced by \cref{constr:lreduction} contain two vertices for every vertex of the \textsc{Independent Set} instance and one vertex for every edge of the \textsc{Independent Set} instance. Further, we set~$\Delta=2$ and~$k=\ell\cdot h_1 + \sum_{i\in[\ell]}|F_i|$.

This instance can be constructed in polynomial time and~$|V|$ is polynomially upper-bounded by the maximum size of an input instance. It is easy to check that the extra edgeless snapshot contained in each constructed temporal graph $\G^{(i)}$ prevents the $\Delta$-temporal matchings from two adjacent constructed graphs $\G^{(i)}$ and $\G^{(i+1)}$ for $i\in[\ell-1]$ to interfere, that is, matching two vertices with a time edge from $\G^{(i)}$ cannot block vertices from $\G^{(i+1)}$ from being matched. Furthermore, since we assume that no instance $(H_i,h_i)$ of \textsc{Independent Set} contains an independent set of size $h_1+1$, it cannot happen that the $\Delta$-temporal matching of a constructed temporal graph $\G^{(i)}$ is larger than $h_1+|F_i|$. It follows from the proof of \cref{th:apx} that the constructed \textsc{Temporal Matching} instance is a 
\yes-instance if and only if for every~$i\in [n]$ the \textsc{Independent Set} instance $(H_i,h_i)$ is a \yes-instance.

Since \textsc{Independent Set} is NP-hard under the above described restrictions~\cite{garey1974some} and we AND-cross-composed it into \textsc{Temporal Matching} with $\Delta=2$ parameterized by $|V|$, this proves the proposition.
\end{proof}

\section{Conclusion}

The following issues remain future research challenges.
First, on the side of polynomial-time approximability,
improving the constant approximation factors is desirable and seems
feasible. Beyond, by lifting polynomial time to FPT~time, even approximation
schemes in principle seem possible, thus circumventing our APX-hardness result
(\cref{th:apx}).
Taking the view of parameterized complexity analysis in order to cope
with NP-hardness, a number of directions are naturally coming up.
For instance, based on our fixed-parameter tractability result for the
parameter solution size $k$, 
the question for a polynomial-size kernel for parameter $k$ naturally arise.
To enlarge the range of promising and relevant
parameterizations,
one may extend the parameterized studies to
structural graph parameters combined with $\Delta$ or the lifetime of the temporal graph.
In particular, treedepth combined with $\Delta$ is left open,
since it is a ``stronger'' parameterization than the one used in
\cref{thm:fpt-for-alpha-delta} but unbounded in all known NP-hardness reductions.

\bibliographystyle{siamplain}
\bibliography{references}
\end{document}